\documentclass{IEEETran}
\usepackage{balance}

\usepackage{microtype}
\usepackage{booktabs} 

\usepackage{graphicx}

\usepackage{epstopdf}
\usepackage{subcaption}
\usepackage{fixltx2e}
\usepackage{balance}
\usepackage[linesnumbered,ruled,vlined,noend]{algorithm2e}
\usepackage{textcomp}
\usepackage{amsmath}

\usepackage{amssymb}  
\usepackage{amsfonts}
\usepackage{color}

\usepackage{xspace}
\usepackage{multirow}
\usepackage{stfloats}
\usepackage{footnote}
\makesavenoteenv{tabular}
\makesavenoteenv{table}

\usepackage{slashbox}
\usepackage{enumitem}
\usepackage{amssymb}
\usepackage{wrapfig}

\usepackage{amsthm} 
\newtheoremstyle{dotless}{3pt}{3pt}{\itshape}{}{}{}{}{}

\usepackage{titlesec}
\titlespacing{\section}{0pc}{1.5pc}{1pc}

\newcommand{\ignore}[1]{}
\newcommand{\stitle}[1]{\vspace{1ex}\noindent{\bf #1}}
\newcommand{\sstitle}[1]{\vspace{1ex} \noindent{\underline{#1}}}

\newcommand{\ie}{\textit{i}.\textit{e}., }

\newcommand{\eg}{\textit{e}.\textit{g}., }

\newcommand{\etc}{\textit{etc}.}

\newcommand{\cf}{\textit{c}.\textit{f}.{ }}

\newtheorem{exmp}{\textbf{Example}}
\newtheorem{proposition}{\textbf{Proposition}}
\newtheorem{obs}{\textbf{Observation}}
\newtheorem{defn}{\textbf{Definition}}

\newcommand{\eat}[1]{}
\usepackage{todonotes}

\usepackage[font={small},aboveskip=1pt]{caption}
\usepackage{lipsum}
\usepackage[htt]{hyphenat}

\begin{document}
\begin{sloppypar}
\title{Finding Minimum Connected Subgraphs with Ontology Exploration on Large RDF Data}

\author{
\IEEEauthorblockN{Xiangnan Ren$^{\ast 1}$,  Neha Sengupta$^{\ast 2}$, Xuguang Ren$^{\ast 3}$, Junhu Wang$^{\dagger}$, Olivier Curé$^{\star}$} \\
\IEEEauthorblockA{
$^{\ast}$\small{Inception Institute of Artificial Intelligence, UAE} \\
\small{$^{\ast}$\{$^1$xiangnan.ren,$^2$neha.sengupta,$^3$xuguang.ren\}@inceptioniai.org}} \\
\IEEEauthorblockA{
$^{\dagger}$\small{Griffith University, Queensland, Australia} \\
\small{$^{\dagger}$j.wang@griffith.edu.au}} \\
\IEEEauthorblockA{
$^{\star}$\small{Université Paris-Est, Marne-la-Vallée, France} \\
\small{$^{\ast}$olivier.cure@u-pem.fr}}
\\\vspace*{-0.7cm}
}
\thispagestyle{empty}
\pagestyle{empty} 
\maketitle

\begin{abstract}
In this paper, we study the following problem: given a knowledge graph (KG) and a set of input vertices (representing concepts or entities) and edge labels,  we aim to find the smallest connected subgraphs containing all of the inputs. This problem plays a key role in KG-based search engines and natural language question answering systems, and it is a natural extension of the Steiner tree problem, which is known to be NP-hard. We present RECON, a system for finding approximate answers. RECON aims at achieving high accuracy with instantaneous response (i.e., sub-second/millisecond delay) over KGs with hundreds of millions edges without resorting to expensive computational resources. Furthermore, when no answer exists due to disconnection between concepts and entities, RECON refines the input to a semantically similar one based on the ontology, and attempt to find answers with respect to the refined input. We conduct a comprehensive experimental evaluation of RECON. In particular we compare it with five existing approaches for finding approximate Steiner trees. Our experiments on four large real and synthetic KGs show that RECON significantly outperforms its competitors and incurs a much smaller memory footprint.
%
\end{abstract} 

\noindent{Keyword Search, Knowledge Graph, Ontology, Reasoning}

\section{Introduction}
\label{sec:introduction}

Knowledge graphs (KGs) have grown tremendously popular over the past decade. This is mainly due to their ability to compactly encode highly connected and diverse information \cite{DBLP:journals/tkde/Hu0YWZ18}. A KG represents entities and relationships between them via vertices and edges. Over the past few years, several very large KGs rich in real world information have emerged with examples including DBPedia~\cite{DBLP:conf/semweb/AuerBKLCI07}, Yago~\cite{DBLP:conf/cidr/MahdisoltaniBS15}, and ConceptNet~\cite{DBLP:journals/corr/SpeerCH16}, among others.
KGs are typically stored in Resource Description Framework (RDF) triple stores, \ie lists of (subject, predicate, object) triples, in which a predicate describes the relationship a subject entity has with an object entity \cite{ DBLP:journals/tkde/Hu0YWZ18}.

One of the primary applications of KGs is in search engines or question answering systems. These systems typically take a natural language question or keywords as input, then extract the corresponding entities and relations via natural language processing (NLP) techniques and map them into vertices and edge labels on a KG. The models thus search for small subgraphs covering all such vertices and edge labels for answer set retrieval \cite{DBLP:journals/vldb/ZouCOZ12,DBLP:journals/is/PengZQ17,DBLP:conf/sigmod/ZouHWYHZ14,DBLP:journals/tkde/Hu0YWZ18,DBLP:conf/icde/KasneciRSSW09}. 

On large KGs, finding such subgraphs is not an easy task. For greater relevance, the returned subgraphs should be minimal, or certain other criteria should be maximized~\cite{DBLP:conf/sigmod/HeWYY07,DBLP:journals/tkde/Hu0YWZ18}. To meet real-time response requirements, the system should produce an answer set with reasonable latency \cite{DBLP:conf/icde/YangAJTW19}. Ideally, the system should have some reasoning capability so that it can return subgraphs that contain semantically similar entities if there is no connected subgraph containing the original input vertices/edge labels.

\begin{figure}[h]
\begin{center}
\includegraphics[width=0.8\linewidth]{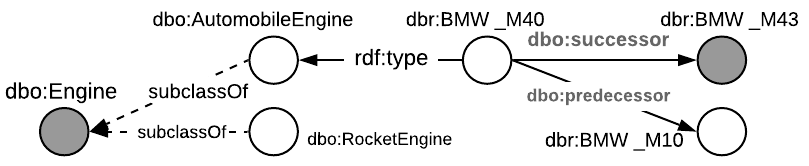} 
\captionof{figure}{An example of a knowledge graph for ontology-based reasoning. Dotted lines represent ontological relationships while solid lines represent instance-level relationships. }
\label{fig:reasoning_semantic}
\end{center}
\end{figure}
\vspace{-1mm}

In this paper, we focus on finding minimum connected subgraphs (MCSs) with ontological reasoning over RDF data. The problem can be formalized as a special case of keyword search on graphs: given a set of input keywords, where each keyword is mapped into either a unique vertex or an edge label, our aim is to return the minimum connected subgraphs which cover all the input keywords. If such an MCS does not exist, we aim to return the MCSs that cover semantically similar keywords, where semantic similarity is derived from the ontology. We use the following example to explain the effect of ontological reasoning. 

\begin{exmp}
Figure \ref{fig:reasoning_semantic} shows a KG with both ontological and instance-level relationships between entities. The input query \{\small{\texttt{dbr:BMW\_M43, dbo:Engine, dbo:successor, dbo:predecessor}}\} could be parsed from \{\small{\texttt{BMW M43, Engine, Successor, Predecessor}}\} by an NLP layer, intending to retrieve the predecessor of the engine which was succeeded by BMW M43, will return an empty result due to the disconnection between these vertices at an instance level. However, ontology-based reasoning enables \small{\texttt{dbo:Engine}} to be refined to \small{\texttt{dbo:AutomobileEngine}}, revealing that \small{\texttt{dbr:BMW\_M10}} is the predecessor of \small{\texttt{dbr:BMW\_M40}}. 
\end{exmp}

The task of finding an MCS is an extension of the problem of finding a Steiner tree (ST), thus it is computationally intractable. \ignore{when the input keywords do not have edge labels, an MCS is reduced to a Steiner tree.} 
In fact, given a set of input keywords, to find an MCS, one can first find an ST with respect to the matched vertices, and then extend the ST to an MCS by adding edges that contain the input edge-labels (and which is our approach). Therefore our problem is closely related to algorithms for finding (approximate) STs. While there are algorithms for finding STs (e.g., \cite{DBLP:conf/icde/DingYWQZL07}, \cite{DBLP:conf/sigmod/LiQYM16}), they do not scale to very large graphs. Existing works on finding approximate STs, on the other hand, are either too slow or cannot return accurate results. We are not aware of systems that have ontological reasoning as discussed above.

\stitle{Our Approach and Contribution.} We present RECON, a framework for finding approximate MCSs with concept-level inference on KGs. RECON achieves high approximation quality in near real-time (milliseconds delay) on large KGs (hundreds of millions of triples), with moderate computational resources.
\begin{itemize}
\item To scale up to large graphs, we propose a novel light-weight sketch index for the computation of Steiner trees. This index can be constructed much more quickly and it takes much less space than those in previous work.
\item To retain high accuracy, we propose a method to ``patch-up" the sketch index with pruned landmark labeling index (PLL) \cite{DBLP:conf/sigmod/AkibaIY13} for shortest distance computation. 
\item We propose a method for conducting concept-level inference on KGs to deal with the problem of empty answers. Our method is based on a novel similarity measurement for keyword queries.
\item We conduct a comprehensive evaluation of RECON over real-world and synthetic datasets, and compare it against five existing systems for the computation of (approximate) STs, which lies at the core of our approach.
\end{itemize}

\stitle{Organization} The rest of this paper is organized as follows: Section \ref{sec:Preliminary} introduces preliminary concepts and formally defines the problem. In Section \ref{sec:solution-overview}, we provide an overview of the proposed solution. Section \ref{sec:sketch_construction} presents a fast procedure for sketch construction. Section \ref{sec:scg_construction} explains how to construct the smallest connected subgraphs. Section \ref{sec:ontologyExploration} describes the reasoning module in detail, and demonstrates how to transform an input query into SPARQL. Section \ref{sec:experiments} presents the experimental results. Section \ref{sec:relatedWork} outlines previous work in the literature. Finally, we conclude the paper in Section \ref{sec:conclusion}. 

\section{Preliminaries}
\label{sec:Preliminary}
A KG, or RDF graph, is composed of triples of the form ($subject, predicate, object$). Each triple can be regarded as a labelled edge from $subject$ to $object$ with $predicate$ being the edge label, and represents either a statement about entities or a description of the relationship between two concepts. Accordingly, a $KG$ can be divided into two parts: the entity description part and the ontology part. We formally define a KG using description logic \cite{DBLP:conf/dlog/2003handbook}.  

\begin{defn}[Knowledge Graph]
\label{def:kg}
A knowledge graph is an edge-labeled graph $KG= G \cup T$. The subgraphs $G$ and $T$ correspond to the ABox (assertional box) and the TBox (terminological box), respectively. An ABox contains extensional knowledge, \ie a set of facts, which are represented as triples of the following three forms: $C(\alpha)$ concept assertion (or type assertion), $R(\alpha, \beta)$ role assertion, and $P(\alpha, v)$ attribute assertion, where $C$ is an entity type (a.k.a. a concept), $\alpha$, $\beta$ are entities, and $v$ is the value of an entity's attribute. A TBox contains intensional knowledge and describes the general properties of concepts and relations. For the TBox, we focus on type inclusion assertions of the form $C_1 \sqsubseteq C_2$, where $C_1$, $C_2$ are concepts, and $\sqsubseteq$ denotes a concept subsumption.
\end{defn}

Intuitively, the ABox and TBox correspond to the entity description part and the ontology part, respectively. For example, in Figure \ref{fig:reasoning_semantic}, the edge ({\tt BMW\_M40, successor, BMW\_M43}) is a role assertion in the ABox, and ({\tt AutoMobileEngine, subclassOf, Engine}) belongs to the TBox. 

\begin{defn}[Keyword Query]
Let $KG = G\cup T$ be a KG. A keyword is a vertex in $KG$ (referred to as a {\em vertex keyword}) or an edge label in $G$ (referred to as an {\em edge label keyword}).
A keyword query is a set of keywords that includes at least one vertex keyword. We use $w_V$ and  $w_{EL}$ to denote the vertex keyword set and the edge label keyword set, respectively. 
\end{defn}

Note that in the general case,  a \emph{keyword} can be any text that matches one or more vertices or edge labels in the KG. In this work, however,  we focus on a special case where each {keyword} can map to only one vertex or edge label in the KG. Also, we assume a keyword query contains at least one vertex keyword. This is because an edge label usually maps to numerous edges in the KG. Therefore, a keyword set containing only edge labels is not very useful in practice.

\begin{defn}[MCS]
Let $KG$= $G\cup T$ be a KG, $w$ be a keyword query, and $G(w)$ be the set of all connected subgraphs of $G$ containing all elements of $w$. We say $g(w) \in G(w)$ is a {\it minimum connected subgraph (MCS)} for $w$ in $G$, if $g(w)$ has the smallest size (total number of edges and vertices) among all of the subgraphs in $G(w)$. 
\end{defn}

Note that we require the MCS to be a connected subgraph of the ABox rather than that of the entire KG. An MCS is an extension of {Steiner tree} (ST). 

\begin{defn} (Steiner Tree)
Given a set of vertices $V'$ in a graph, a Steiner tree (ST) (w.r.t $V'$) is the smallest tree in the graph that contains all the vertices in $V'$. 
\end{defn}

When the input keywords contain no edge label, the MCS becomes an ST. Finding an ST is known to be an NP-hard problem \cite{DBLP:journals/mp/ChopraR94}, therefore finding an MCS is also NP-hard. 

\begin{proposition}
\label{th:scg_np_hard}
For a given keyword query in a KG, finding an MCS is NP-hard.
\end{proposition}
\begin{proof} 
We can easily reduce the ST problem to the MCS  problem. 
Given a set of vertices, we can treat them as the vertex keywords and find an MCS. This MCS must be a tree since otherwise its spanning tree will be smaller and contains all keywords. By definition of an MCS, it must be an ST. 
\end{proof}

\begin{defn}[Dangling Edge Label]
Let $G$ be the ABox of a KG and $w = w_V \cup w_{EL}$ be a keyword query. Let $Tr$ be a tree in $G$ that contains all the vertices in $w_V$.  For each edge label $el \in w_{EL}$, we say $el$ is {\it covered} by $Tr$ if $el$ appears in $Tr$. The edge labels in $w_{EL}$ that are not covered by $Tr$ are called {\em dangling edge labels} w.r.t $Tr$.
\end{defn}

\begin{exmp}
Let $w=$\{\texttt{BMW\_M43, successor, predecessor, AutomobileEngine}\} be a keyword query for the graph shown in Figure~\ref{fig:scg_ex1} (a). Figure~\ref{fig:scg_ex1} (b) shows a non-MCS graph that contains all keywords in $w$. Figure~\ref{fig:scg_ex1} (c) shows an MCS, where \texttt{predecessor} is a dangling edge label w.r.t. the tree $<\texttt{BMW\_M43, BMW\_M40, AutomobileEngine}>$.
\end{exmp} 

\vspace{-2mm}
\begin{figure}[h]
\begin{center}
\includegraphics[width=1\linewidth]{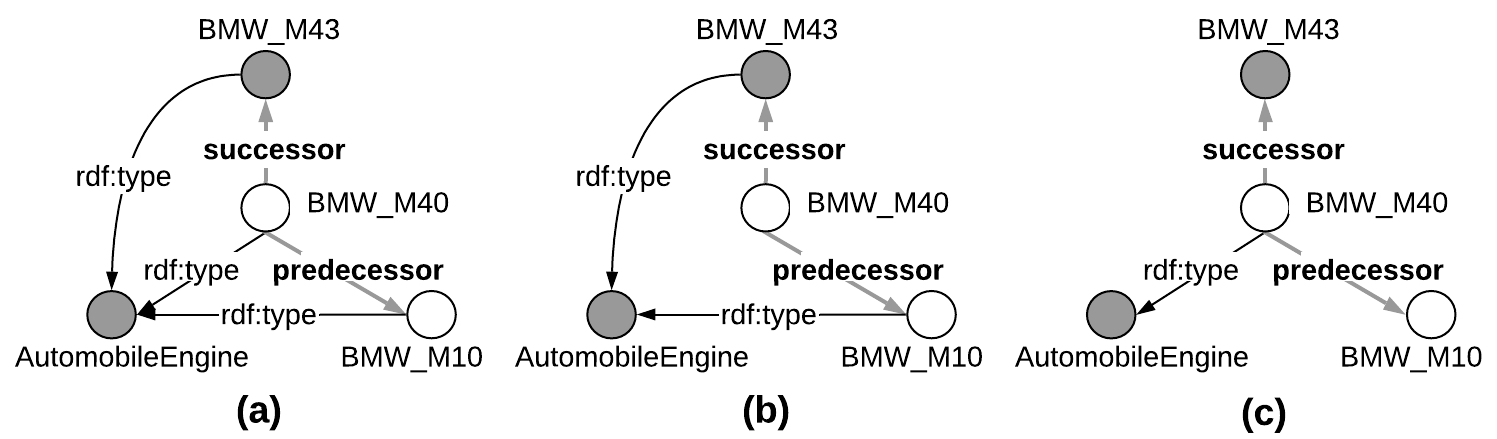} 
\captionof{figure}{Example MCS. Vertices and edge labels marked in gray are the mentioned keywords: (a) $G$ - the ABox; (b) a subgraph of $G$; (c) an MCS of $G$.}
\label{fig:scg_ex1}
\end{center}
\end{figure}

Note that for a keyword set $w$, an MCS may not even exist. In this case, refining the keywords of $w$ to subtypes may generate a relevant answer to the query (see Figure \ref{fig:reasoning_semantic} for an example). As outlined in later sections, the goal of reasoning is to progressively refine the keywords of the query until an MCS can be retrieved. This allows us to return results that are  as relevant to the query as possible.

\stitle{Problem Statement.} Given a keyword query over a KG, our goal is to retrieve all the MCSs that contain either all of the given keywords or, if such an MCS does not exist, those that contain semantically similar keywords.

\vspace{-1mm}
\subsection{Sketch-Based Index}
\label{sec:landmark-based-index}
Our approach to building the MCS involves using a \textit{sketch-based index}.
The sketch-based, or landmark-based index, is widely used for distance query or ST approximation in large graphs \cite{DBLP:conf/cikm/PotamiasBCG09,DBLP:journals/jal/CowenW04,DBLP:conf/wsdm/SarmaGNP10,DBLP:conf/cikm/GubichevN12,DBLP:conf/sigmod/AkibaIY13,DBLP:conf/cikm/TretyakovAGVD11}. It commonly includes two stages: (S1) Select a set of landmark vertices, and precompute the shortest path from each vertex to the landmarks. To do so, a breath-first search (BFS) is applied from the selected landmarks to $v$, and the explored paths are stored in $sk(v)$ (\ie the \emph{sketch} of $v$, which is a tree rooted at vertex $v$ with the shortest paths reaching the landmarks). (S2) Combine the distance or path information obtained in (S1). Using techniques such as triangle inequality and shortcutting, the time complexity of the shortest path approximation can be reduced to $\mathcal{O}(S)$, where $S$ is the size of the landmarks \cite{DBLP:conf/cikm/GubichevBSW10}.

\stitle{Limitations.} The existing approaches for landmark selection require a series of BFS procedures over the entire graph to obtain the path information from each vertex to its landmarks. For example, in \cite{DBLP:conf/cikm/GubichevBSW10,DBLP:conf/wsdm/SarmaGNP10,DBLP:conf/cikm/GubichevN12}, the landmark sets are built by uniformly sampling at random $N_S = log|V|$ seed sets of vertices $S_1,S_2,...,S_{N_S}$, with sizes $1, 2,...2^{N_S - 1}$, respectively. The authors compute the distance from each non-seed vertex to its seed vertex and keep the closest one as the landmark. This process of landmark selection and path computation is repeated $k$ times. Although such an algorithm approximates the shortest path well, its $\mathcal{O}(k|V|(|V|+|E|))$ time cost prohibits it from scaling to large graphs. To overcome this drawback, we propose a new method for the fast construction of a light-weight sketch index in Section \ref{sec:sketch_construction}. 

\subsection{Pruned Landmark Labeling (PLL)}
\label{sec:PLL}
The pruned landmark labelling index is a two-hop labeling index for the fast computation of shortest distances/paths \cite{DBLP:conf/sigmod/AkibaIY13}. However, PLL does not provide alternative shortest paths. Therefore, it is not suitable for searching STs or MCSs on its own. Further, with large dense graphs, the PLL index is time-consuming to construct, and it takes too much space. Therefore, we combine an {\em $r$-restricted} version of PLL with the sketch index in our search for STs. The $r$-restricted version of PLL is obtained by limiting the BFS depth to a small constant $r$ during index construction. Note that such a PLL index can still find the correct shortest distance between two nodes if their distance is  no more than $r$.               

\section{Solution Overview}
\label{sec:solution-overview}

 Algorithm \ref{algo:solution-overview} outlines the generic framework of RECON, which consists of two parts, \ie offline preprocessing and online computation.
 

\vspace{-2mm}
\begin{algorithm}[ht!]
\small
\caption{Generic Framework}
\label{algo:solution-overview}
\KwIn{Knowledge graph $KG = G \cup T$, keyword set $w$}
\KwOut{Answer for $w$ on $KG$} 
$SK \gets $ generateSketch($G$) \\
$PLL \gets $ generatePLL($G$)   \\
$ST(w) \gets $ generateST($w, SK, PLL$)   \\
\If{$ST(w) = \emptyset$}{
  Refine $w$ using $T$ \\
  $ST(w) \gets $ generateST($w$, $SK$, $PLL$, $T$)
} 
$MCS(w) \gets $ generateMCS($ST(w), G$)  \\
$q^w \gets$ generateSPARQL($MCS(w)$)  \\
$ans \gets$ execute $q^{w}$ in triple store         \\
{Reformat} $ans$ as MCSs \\
\Return $ans$ 
\end{algorithm}
\vspace{-3mm}
 
 \stitle{Offline Preprocessing.} During offline preprocessing (lines 1 to 2), we construct the sketch index and PLL index as described in Sections~\ref{sec:landmark-based-index} and \ref{sec:PLL}. These indexes will be used to build an approximate ST w.r.t the input vertex keywords during the online computation stage.

\stitle{Online Computation.} 
We first try to construct an approximate ST with respect to the vertex keywords using the sketch and PLL indexes (line 3). We stress here that the edge labels in the keyword set $w$ are also considered during the ST construction so as to minimize the number of dangling edge labels. If an ST does not exist due to disconnection between the vertices, we will refine the keywords using the ontology and try to generate the STs again (lines 4-6). Based on the approximate STs generated, we construct the approximate MCSs by handling dangling edge labels (line 7).  Since an MCS naturally represents the algebra of a SPARQL query, we can generate a SPARQL query from the MCSs \cite{sparqlProtocol} (line 8). Then we execute the SPARQL query in a triple store and reformat the query results to produce more MCSs (line 9-11). The advantage of using SPARQL is that we can profit from the highly optimized triple store for faster answer retrieval, compared with repeated search over the KG. Note that the result reformatting is done by appending the output tuples of the SPARQL query to an obtained MCS pattern. Each tuple produces an MCS, which is flushed immediately into the disk on-the-fly with negligible cost in time and memory storage.

\section{Sketch Construction}
\label{sec:sketch_construction}

To enhance the scalability to large graphs, we propose a novel approach for landmark selection and sketch construction in Algorithm \ref{algo:sketch-build}.

\vspace{-2mm}
\begin{algorithm}[ht!]
\small
\caption{Offline Sketch Construction}
\label{algo:sketch-build}
\KwIn{$G = (V,E)$, radius $r$, integer $k$}
\KwOut{$\{sk(v)\}, \forall v \in V$} 

\For{i=1 \text{to} k}{
Mark all nodes as unvisited \\
\While{$\exists v \in V$ which is unvisited}{
$l \gets $ selectLandmark($V$), $l$ is unvisited \& unused \\
do BFS from $l$ within radius $r$\\
\For{each node $v$ visited during BFS}{
/* add shortest path to sketch */\\
$sk(v) \gets sk(v) \cup sp_{l,v}$ \\
mark $v$ as visited \\
}
mark $l$ as used
}
}
\Return $\{sk(v)\}, v \in V$
\end{algorithm}
\vspace{-3mm}
\stitle{\color{red}} As shown in Algorithm~\ref{algo:sketch-build}, we first select a landmark $l$ using the function selectLandmark($V$). As will be explained shortly, the function  prioritizes vertices with high \emph{informativeness} (see Definition \ref{def:selectivity}).\ignore{a measure of the diversity of information stored at this vertex)} Then, we perform BFS from the selected landmark $l$ within a radius $r$ (\ie we go at most $r$ hops from $l$) and update the sketch of all reached vertices (lines 5 to 9). Note that we do not traverse the entire graph from $l$. Once all the $r$-hop neighbors of $l$ have been visited, we pick another unused landmark $l'$ that has not yet been visited, and we continue this process until the whole graph has been visited. This procedure ensures that every vertex has exactly one landmark. We repeat this process for $k$ times, where $k \ll |V|$. Compared with existing approaches like \cite{DBLP:conf/cikm/PotamiasBCG09,DBLP:journals/jal/CowenW04,DBLP:conf/wsdm/SarmaGNP10,DBLP:conf/cikm/GubichevN12,DBLP:conf/sigmod/AkibaIY13,DBLP:conf/cikm/TretyakovAGVD11}, the time complexity of sketch construction is reduced from $\mathcal{O}(k|V|(|V| + |E|))$ to $\mathcal{O}(k(|V| + |E|))$, where $k$ refers to the maximum number of landmarks for $sk(v)$. Note that we will not choose a landmark that has been used in previous rounds (line 4). It is not hard to see that a larger $k$ means a higher probability that two vertices will share a common landmark. In our experiments, we use $k = log|V|$. Thus, the space complexity for sketch storage is reduced to $\mathcal{O}(|V|log|V|)$, considering $r$ as a constant $\ll|V|$. 

\stitle{Landmark Selection.} Instead of uniformly selecting landmarks at random, we select them by weighted reservoir sampling (A-Res algorithm \cite{DBLP:journals/ipl/EfraimidisS06}). Inspired by \cite{DBLP:conf/icde/YangAJTW19}, we assign each vertex a weight or score, called informativeness, and choose the node with the highest informativeness from the unvisited and unused nodes as the next landmark.

\begin{defn}[Informativeness]
\label{def:selectivity}
$\forall v \in V$, let $EL(v)$ be the set of unique edge labels incident on $v$. We call
\begin{small}
\begin{align*}
 I(v) =  log|EL(v)| * log(deg(v))
\end{align*}
\end{small}
the informativeness of $v$, where $deg(v)$ is the degree of $v$. We assign $I(v)$ as the weight of $v$ which determines its probability of being selected as a landmark.
\end{defn}
\vspace{-2mm}

Intuitively, the informativeness is determined by the number of distinct edge labels incident on the node as well as the degree. It is possible to have $deg(v) \gg |EL(v)|$, so we normalize the degree and the number of edge labels by $log$ scale\ignore{\cite{DBLP:journals/jis/FernandezMRG18}}.

There are two reasons for introducing $I(v)$: (i) the conventional way to select landmarks is to choose high-degree vertices, since high-degree vertices means high centrality, and choosing such vertices as landmarks is likely to produce good shortest path estimation. However, for keyword search we would like the vertices connected by the landmarks to be meaningfully related, and while vertices connected by very high-degree nodes may be completely unrelated to each other. For instance,  vertices connected by the high-degree vertex {\tt owl:Thing}. Therefore, using degree alone is not the best strategy for our problem.  
(ii) Our system involves the search of edge labels, so we aim to retain the edge label diversity in the sketches. 

\stitle{Sketch Balancing.} Due to the randomness of the weighted sampling process, it is possible for the distribution of the path information from different categories  of assertions (\ie role, concept and attribute) to be unbalanced. In the worst case, $sk(v)$ may only contain the literal assertions. Therefore, $v$ to any other entities or type vertices becomes unreachable in the sketches. To address this issue, we split the graph $G$ into three subgraphs based on Definition \ref{def:kg}, \ie $G = G_{R \mapsto R} \cup G_{R \mapsto T} \cup G_{R \mapsto L}$ ($G$ is equal to the union of role assertions ($G_{R \mapsto R}$), type assertions ($G_{R \mapsto T}$), and attribute assertions $(G_{R \mapsto L})$). Finally, we perform the sketch construction on each subgraph independently and aggregate the results together. 

\stitle{Weakness of Sketches.} As mentioned previously, MCS construction relies on the shortest path estimation on sketches. However, we have opted for a trade-off between algorithm scalability and sketch quality, leading to the following observations.

\begin{figure}[h]
\begin{center}
\includegraphics[width=0.9\linewidth]{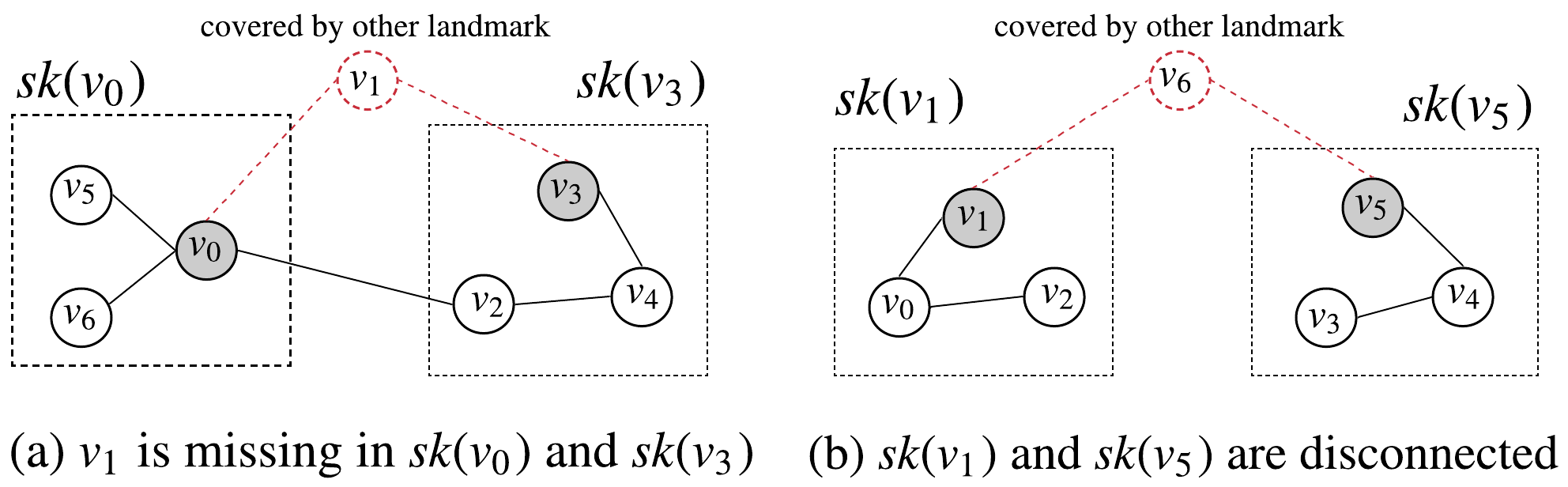}
\captionof{figure}{Examples of inaccurate shortest path estimation.}
\label{fig:failed-sketch}
\end{center}
\end{figure}

\begin{obs}
\vspace{-1mm}
\label{sketch_low_quality}
Sketches generated by Algorithm \ref{algo:sketch-build} may yield inaccurate shortest path estimation. This is because two sketches $sk(v)$ and $sk(v')$ might not share any intersection, or the shortest path between $v$ and $v'$ cannot be found using their sketches.
\end{obs}

\begin{exmp}
Sketches shown in Figure \ref{fig:failed-sketch} can occur due to the fact that we do not visit vertices already covered by other landmarks. In Figure \ref{fig:failed-sketch} (a), the path between $v_0$ and $v_3$ obtained using the two sketches is $(v_0, v_2, v_4, v_3)$, which is of length 3. However, a shorter path $(v_0, v_1, v_3)$ of length 2 exists in the original graph. In Figure \ref{fig:failed-sketch} (b), $v_6$ appears in neither $sk(v_1)$ nor $sk(v_5)$, since $sk(v_1)$ and $sk(v_5)$ are disconnected. Therefore, the path $(v_1,v_6,v_5)$ cannot be found using the sketches.
\end{exmp}

Moreover, since we limit the radius of the BFS to $r$, if two vertices have a distance larger than $2r$, then we cannot accurately estimate the distance using the sketches.

To address the above-mentioned weaknesses, we build a PLL index to iteratively add shortest paths to the sketches on-the-fly. We call this process sketch {\em patch-up}. The details will be given in the next section.

\section{MCS Construction}
\label{sec:scg_construction}
We focus on the construction of the initial MCSs (line 3, 7, Alg.~1) in this section. The process consists of the following three steps: (i) we patch up the sketches using PLL; (ii) we construct an approximate ST using the patched-up sketches, choosing shortest paths containing more input edge labels whenever possible; (iii) we generate the approximate MCS by inserting edges to cover dangling edge labels. 

Before explaining the sketch patch-up, we define the concept of {\em vertex occurrence}.

\begin{defn}[Vertex Occurrence]
\label{def:occ}
Given a vertex $v_i$ and the sketches of keyword vertices $SK=\{sk(t_1),...,sk(t_n)\}$, the occurrence of $v_i$ in $SK$, denoted  $occ_i$,  is the number of sketches in $SK$ that contains $v_i$, that is, $occ_i=|\{ t_j ~|~ v_i\in sk(t_j) \}|$.
\end{defn}

We use the \emph{vertex occurrence} to approximate the central vertex (or Jordan center \cite{DBLP:books/cu/WF1994}) in our algorithm. The Jordan center of a graph $J(G)$ is defined as the set of vertices that have the smallest eccentricity, \ie $\forall v \in J(G)$, the largest distance from $v$ to other vertices of $G$ is minimal. 
Our motivation is as follows:
\begin{obs}
\label{obs:occ}
During the ST construction, the higher the occurrence of the vertex, the more likely it will appear in the ST.
\end{obs}

\subsection{Sketch Patch-Up}
\label{subsect:sketch-patch-up}

We use two types of patch-ups: keyword-keyword patch-up and central vertex keyword patch up, as explained below.

\stitle{Keyword-Keyword (KK) Patch-Up.} KK patch-up inserts shortest paths from one keyword vertex to other keyword vertices. The process is described in Algorithm \ref{algo:sketch_init} (lines 2-10). We start with the loading of sketches. Lines 2-4 initialize the $SearchState_i$ for each keyword $t_i \in w_V$. During BFS, the $SearchState_i$ maintains: $G_{sk}(t_i)$, the graph generated from $sk(t_i)$, $G_{sk}(t_i)$, which is progressively updated during MCS construction; the queue of BFS for $t_i$; the visited vertices $F_i$ (an ordered hash set). After that, we extract the shortest paths from the PLL index between all vertex keyword pairs and insert them into the sketches (lines 6 to 10). \ignore{Figure \ref{fig:sketch-patch-up} II gives an example of path insertion in KK patch-up.}

\begin{exmp}
Consider the input vertex keywords $t_0$, $t_1$, $t_2$ and the corresponding sketches $sk(t_0)$, $sk(t_1)$, $sk(t_2)$ shown in Figure \ref{fig:sketch-patch-up}.I. Figure \ref{fig:sketch-patch-up}.II gives an example of path insertion in KK patch-up, \ie edges $e_{0, 14}$, $e_{14, 1}$ and $e_{6,2}$ are inserted into $G_{sk}(t_0)$. Figure \ref{fig:sketch-patch-up}.III.(a) shows the union of sketches after KK patch-up, $v_{10}$, $e_{2,10},e_{10,1},e_{0,14}$ are inserted (we assume the paths $(t_0,v_{14},t_1), (t_0,v_6,t_2), (t_1,v_{10},t_2)$ are found using the PLL index between the keyword vertices $t_0,t_1$ and $t_2$).
\end{exmp}

KK patch-up alone is not sufficient for handling the types of failures illustrated in Figure \ref{fig:failed-sketch}. This is mainly due to the infeasibility of PLL computation in large dense graphs. For example, in DBpedia, when the number of hops for BFS in PLL reaches three and the procedure iterates over only 270 vertices, the index size exceeds the memory limit (60G). As such, we set the number of hops $r$ of BFS in PLL construction to be the same as for sketch construction.

\begin{figure}[h]
\vspace{-1mm}
\begin{center}
\includegraphics[width=0.75\linewidth]{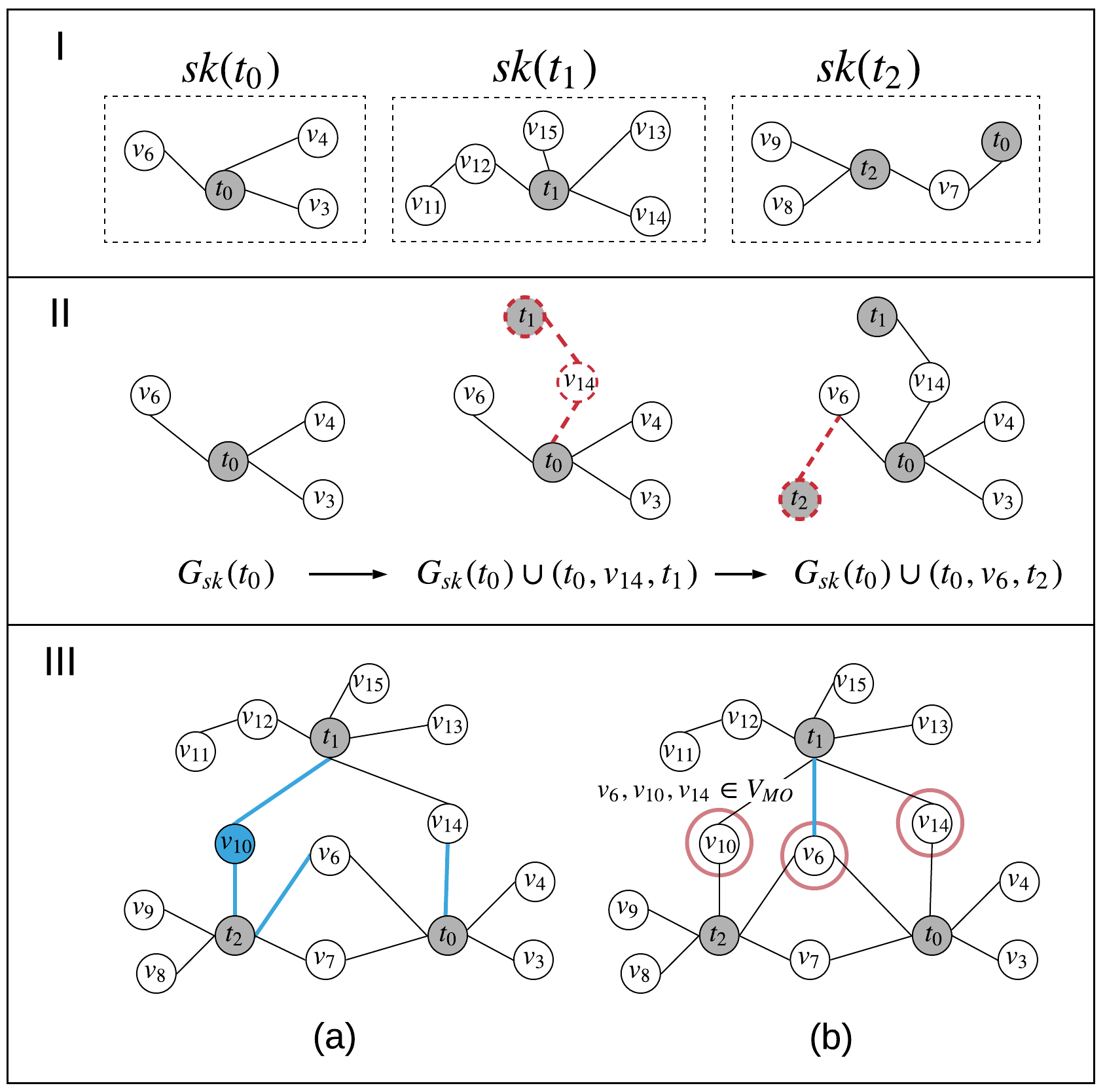} 
\captionof{figure}{Sketch patch-up example. I. Input sketches $sk(t_0)$, $sk(t_1)$ and $sk(t_2)$. II. Path insertion of KK patch-up for $G_{sk}(t_0)$. III. (a) Union of sketches after KK patch-up; (b) union of sketches after CK patch-up. Blue vertices/edges are inserted after patch-up. 
} 
\label{fig:sketch-patch-up}
\vspace{-2mm}
\end{center}
\end{figure}

\stitle{Central Vertex-Keyword (CK) Patch-Up.} CK patch-up inserts the paths from non-keyword vertices with higher occurrences in sketches into the keyword vertices. A vertex with higher occurrence is more likely to become the center of an ST, since such a vertex is expected to minimize the largest distance to the keyword vertices. We use CK patch-up to reveal the shortest paths from a non-keyword vertex that exists at the intersection of the sketches to the vertex keywords. For example, CK patch-up may be used to discover the shortest path from a vertex $v_k$ that has high occurrence in the sketches to keyword vertices $t_i$ and $t_j$. Thereafter, the shortest path between $t_i$ and $t_j$, $sp_{t_i,t_j}$, can be estimated by the concatenation of $sp_{t_i, t_k}$ and $sp_{t_k, t_j}$ (\ie $sp_{t_i, t_k} \oplus sp_{t_k, t_j}$).

CK patch-up is described in lines 11-20 of Algorithm \ref{algo:sketch_init}. We track the occurrence of each vertex in the union of sketches by a map $\mathcal{M}_{occ}$ ($\mathcal{M}_{occ} = \{(v_i, occ_i) | v_i \in SK \}$), and the vertices with highest occurrence (denoted by $V_{MO}, \forall v \in V_{MO}, v \notin w_V$) are extracted by \emph{getMaxOccVertices}. We repeat the CK patch-up procedure until the convergence condition is achieved on the $\tau$-th ($\tau \geq 1$) iteration (Algorithm \ref{algo:sketch_init}, lines 14 to 20):

\begin{small}
\vspace{-3mm}
\begin{align}
\label{cond:iter_patch}
\begin{cases} 
\exists v_m \in V_{MO}, occ_m^{(\tau)} = |w_V| \textbf{, or } \\ 
\forall v_m \in V_{MO}, occ_m^{(\tau)} - occ_m^{(\tau - 1)} = 0 .
\end{cases}
\end{align}
\vspace{-3mm}
\end{small}

Condition (\ref{cond:iter_patch}) illustrates that CK patch-up is terminated either when $\exists v \in V_{MO}, v$ appears in all sketches, or $\forall v \in V_{MO}$, where the occurrence of $v$ is unchanged at the $\tau$-th iteration. As observed in our experiments, the number of CK patch-up iterations is small (typically $\leq 3$). 

\begin{exmp}
In Figure \ref{fig:sketch-patch-up}.III.(b), after the KK patch-up, $v_{6}$ appears in both $sk(t_0)$ and $sk(t_2)$, \ie $occ_6 = 2$. Analogously, we have $occ_{10}$ and $occ_{14}$ equal to 2. Then, we apply CK patch-up to $v_6$, $v_{10}$ and $v_{14}$, and $e_{6,1}$ is revealed. Finally, $v_6$ appears in all the sketches (\ie $occ_6 = 3$), and the termination condition (\ref{cond:iter_patch}) is achieved.
\end{exmp}

\vspace{-4mm}
\begin{algorithm}[ht!]
\small
\caption{Sketch Patch-up}
\label{algo:sketch_init}
\KwIn{vertex keyword set $w_V = \{t_1,...,t_k\}$, $G = (V,E)$}
\KwOut{Graph generated from sketches of $w_V$}

Load sketches $sk(t_1)$,...,$sk(t_k)$       \;

\ForEach{$t_i \in w_V$ }{
    $G_{sk}(t_i) \gets$ $sk(t_i)$  \\
    Initialize  $SearchState_i$    \\
    }
\text{/* Keyword-Keyword Patch-up */} \\
\For{$i \gets 1$ to $k - 1$}{
    \For{$j \gets i + 1$ to $k$}{
        $sp_{i,j} \gets$ retrieveShortestPathByPLL($t_i$, $t_j$)               \\
        $G_{sk}(t_i) \gets G_{sk}(t_i) \cup  sp_{i,j}$      \\
        $G_{sk}(t_j) \gets G_{sk}(t_j) \cup  sp_{j,i}$      \\
    }
}

\text{/* Central-Keyword Patch-up */} \\
$\mathcal{M}_{occ} \gets$ build from  $G_{sk}(t_1),...,G_{sk}(t_k)$ \\
$V_{MO} \gets$ getMaxOccVertices($\mathcal{M}_{occ}$)

$\tau \gets 1$; $occ_{m}^{1} \gets occ_m;~ occ_{m}^{0} \gets 0, ~ \forall v_m\in V_{MO}$  \\
\While{Condition (\ref{cond:iter_patch}) is not satisfied}{
\ForEach{$v_m \in V_{MO}$ }{
    \ForEach{$t_i \in w_V$ }{
     $sp_{i,m} \gets$ retrieveShortestPathByPLL($t_i$, $v_m$)      \\
     $G_{sk}(t_i) \gets G_{sk}(t_i) \cup  sp_{i,m}$         \\ 
     update $\mathcal{M}_{occ}$ 
     }
  }
$V_{MO} \gets$ getMaxOccVertices($\mathcal{M}_{occ}$)  \\
$\tau \gets \tau + 1$
}
\Return $G_{sk}(t_1),...,G_{sk}(t_k)$
\end{algorithm}
\vspace{-2mm}

\subsection{ST Construction}
\label{sec:scg}

For the construction of the MCS we begin by building the ST of $w_V$, which requires a greedy local search algorithm. We start from a keyword vertex $t_i$, conduct BFS locally on the sketches of an keyword vertices, and try to explore the path from $t_i$ to $t_j, t_j \in w_V, i \neq j$. Then we switch to another keyword vertex $t_i'$ through round-robin scheduling and repeat the previous step until the ST is successfully constructed. This local search relies on the following strategies: (i) When multiple shortest paths are discovered from $t_i$ to $t_j$, we choose the one with the highest total occurrence. (ii) The search order between keyword vertices is determined by the \emph{average distance} to the keyword vertices.

Algorithm \ref{algo:steiner_tree_construction} describes the ST construction process. Lines 1 and 2 determine the search order at the initialization stage. We find that the initial search order impacts the compactness of the resulting ST. We start expanding the sketches from the keyword vertex with the smallest degree. The intuition is that, in the greedy BFS search, a vertex with a lower degree is more likely to hit a central vertex earlier (\ie by searching fewer vertices). In Figure \ref{fig:steiner_tree_construction} (a), the search order is determined to be $t_0$, $t_2$, $t_1$.

Before iteration over each $SearchState$, line 3 initializes: an empty ST $G_{st}$; a set $S_{cover}$ which is used to record the keyword vertices covered by $G_{st}$; and the set of dangling edge labels $w_{EL}^d$ (at the beginning, $\forall l \in w_{EL}$, $l$ is a dangling edge label).

\stitle{In each round-robin iteration}, the procedure starts by initializing path map $\mathcal{MP}$ and $V_{cl}$ vertices in the current level (level $l$ stores the $l-$hop neighbors of $t_i$) to be expanded (lines 6 to 8). $\mathcal{MP}$ is a map keyed by keyword vertex pair $<t_i, t_h>, \forall t_i, t_h \in w_V$, which stores all the paths discovered between $t_i$ and $t_h$.

Instead of conducting BFS vertex by vertex as in \cite{DBLP:conf/cikm/GubichevN12,DBLP:conf/icde/KasneciRSSW09}, line 9 expands the BFS in $G_{sk}(t_i), t_i \in w_V$ to process all vertices at the same level (\ie the vertices having the same distance to the source vertex in the BFS procedure). Therefore, multiple shortest paths might be revealed from $t_i$ to $t_h$. In the later process, we apply the \emph{path selection} strategy to decide which path should be kept for ST construction. The algorithm iterates over the $SearchState_h, i \neq h$, to explore the possible paths from $t_i$ to $t_h$ (line 10). In line 11, we verify whether the current vertex $v_j$ and its neighbors are connected to any other already visited vertices by a lookup in the original graph via SPO permutation index \cite{DBLP:journals/pvldb/WeissKB08,DBLP:journals/pvldb/NeumannW08}. If it is the case, the method \emph{getPathFromSketch} extracts paths $sp_{i, j}$ and $sp_{j,h}$ from $G_{sk}(t_i)$ and $G_{sk}(t_h)$, respectively, and concatenates them together (line 12),  which is similar to \cite{DBLP:conf/cikm/GubichevN12}.

The method \emph{getPathsFromMP} extracts all the stored paths between $t_i$ and $t_h$ from $\mathcal{MP}$ and assigns them to $P$ (line 13). If $sp_{i,h}$ (\ie the current revealed path) is shorter than any previously discovered longer path in $P$, we remove $P$ from $\mathcal{MP}$ by \emph{deletePathsFromMP}, and insert $sp_{i,h}$ into $\mathcal{MP}$ (lines 14 to 16) for $<t_i, t_h>$ by \emph{insertPathToMP}. Otherwise, if $|sp_{i,h}|$ is equal to any path between $t_i$ and $t_h$ stored in $\mathcal{MP}$, we directly add it into $\mathcal{MP}$ (lines 19 to 20). Meanwhile, the occurrence of each vertex is continuously updated during the above-mentioned path exploration.

\begin{exmp}
In Figure \ref{fig:steiner_tree_construction} (a), the search process propagates to $v_3$, $v_4$, $v_6$, $v_7$, $v_{14}$,  which are the first-level neighbors of $t_0$ (red dotted circle). By checking the connectivity of $v_3$, $v_4$, $v_6$, $v_7$, $v_{14}$ from already visited vertices, paths $\{t_0, v_{14}, t_1\}$, $\{t_0, v_6, t_1\}$, $\{t_0, v_6, t_2\}$, $\{t_0, v_7, t_2\}$ are discovered and inserted into $\mathcal{MP}$.
\end{exmp}

\vspace{-1mm}
\stitle{Round-Robin Scheduling.} We schedule the iteration over the union of sketches by weighted round-robin \cite{DBLP:journals/jsac/KatevenisSC91} (line 5). The next keyword for which the BFS should expand is determined based on the \emph{average distance} from a keyword vertex $t_i$ to all other keyword vertices as: 

\begin{defn}[Average Distance]
$\forall t_i \in w_V $, the average distance $ \bar{d_i} = \frac{\sum_{j, i \neq j, sp_{i,j} < \infty} |sp_{i,j}|}{|\{v_j | i \neq j, sp_{i,j} < \infty\}|}$.
\end{defn}

 After each iteration over the sketch of  $t_i \in w_V$, we look up the shortest paths in $\mathcal{MP}$ and compute $\bar{d}$ for all keywords to select the keyword to be explored with the smallest average distance in the next iteration (line 5).


\stitle{Path Selection.} At the end of the exploration over each level from a source vertex, we retain only one optimal path to other vertices for further expansion in the next level (line 20). If there are multiple paths of shortest length, we need to select an optimal path to retain (Algorithm \ref{algo:steiner_tree_construction}, Procedure \emph{PathSelection}).

The main goal of path selection is to ensure \emph{structural} and \emph{semantic} correctness. The principle is to maximize semantic correctness while ensuring the most compact structure possible. Lines 2 to 4 of the $PathSelection$ procedure iterate over all paths in $\mathcal{MP}$ by vertex pair.

Based on Definition \ref{def:occ}, we compute the total number of occurrences for a given path (line 5). For each $sp_{i,h} \in P$, we compute $pathOcc$, \ie the total number of occurrences of $sp_{i,h}$, by $getOcc$, where $pathOcc(sp_{i,h}) = \sum_j occ_j,\forall v_j \in sp_{i,h}$. In line 6, $getCoveredEL$ extracts the edge labels ${e}$ covered by $sp_{i,h}$, where $e \in w_{EL}^d$. The path with the highest occurrence score is selected (lines 9 to 11). The motivating intuition is that such a path has a higher chance of being the center of $G_{st}$, which has the minimal largest distance to other vertex keywords.

\begin{algorithm}[ht]
\small
\caption{Steiner Tree Construction}
  \SetAlgoLined\DontPrintSemicolon
\label{algo:steiner_tree_construction}
\KwIn{Graph sketch $G_{sk}(v_i)$, search state $SearchState_i$, $t_i \in w_V$, set of edge labels $w_{EL}$}
\KwOut{ST $G_{st}(w_V)$ covering $w_V$}
\SetKwFunction{SortBySearchPriority}{SortBySearchPriority}
\SetKwFunction{getVisited}{getVisited}
$SearchStates \gets$ new Array $\{SearchState_i$\}                  \\
sortBySearchPriority($SearchStates$)                   \\
$S_{cover} \gets \emptyset$; $G_{st} \gets \emptyset$; $w_{EL}^d \gets w_{EL}$         \\

/* Select $SearchState_i$ based on $\bar{d}$ */  \\
$SearchState_i \gets$ \text{Pick from} $SearchStates$ \\
$\mathcal{MP} \gets \emptyset$ \\

$V_{cl} \gets SearchState_i$.getVerticesByLevel()    \\
$SearchState_i.F_i \gets SearchState_i.F_i \cup V_{cl}$  \\
\ForEach{$v_j \in V_{cl}$}{
\ForEach{$t_h \in SearchState_h$, $h \neq i$}{
\If{$\mathcal{N}(v_j) \cap SearchState_h.F_h \neq \emptyset$ \textbf{or} $v_j \in G_{sk}(t_h)$}{
    $sp_{i, h} \gets$ getPathFromSketch($t_i, v_j, G_{sk}(t_i)$) $\oplus$ getPathFromSketch($t_h, v_j, G_{sk}(t_h)$) \\
    $P \gets $getPathsFromMP($<t_i,t_h>, \mathcal{MP}$). \\
    \If{$P \neq \emptyset$ and $|sp_{i, h}| \le |p|, \exists p \in P$}{
        \text{deletePathsFromMP}($<t_i,t_h>,\mathcal{MP}$)  \\
        \text{insertPathToMP}($sp_{i, h}, <t_i,t_h>,\mathcal{MP}$)  \\
        $S_{cover} \gets S_{cover} \cup \{t_i, t_h\}$           \\
    }
    \ElseIf{$|sp_{i, h}| =  |p|, \forall p \in P $}{
     \text{insertPathToMP}($sp_{i, h}, <t_i,t_h>,\mathcal{MP}$)  \\
    }
}}}
\text{PathSelection}($G_{st}, \mathcal{MP}, SearchStates,w_{EL}^d$)\text{} \\
\If{$|S_{cover}| = |w_V|$}{
    \Return $G_{st}$
}
\setcounter{AlgoLine}{0}
\SetKwProg{myproc}{Procedure}{}{}
\myproc{PathSelection($G_{st}, \mathcal{MP}, SearchStates, w_{EL}^d$)}{
\ForEach{$<t_i, t_h> \in \mathcal{MP}$\text{.keys()}}{
$p_{opt} \gets \emptyset$; $maxOcc \gets 0$; \\
\ForEach{$sp_{i,h} \gets$ getPaths($<t_i, t_h>, \mathcal{MP}$)}{ 
$pathOcc \gets$ getOcc($sp_{i,h}$, $SearchStates$)   \\
$w_{EL}^{cover} \gets$ getCoveredEL($sp_{i,h}$,$SearchStates$) \\
\If{ $sp_{i,h} = \emptyset$}{ \textbf{continue} }
\If{$maxOcc < pathOcc$}{
        $maxOcc = pathOcc$        \\
        $p_{opt}= sp_{i,h}$ } 
\ElseIf{$maxOcc = pathOcc$ and CompareCover($p_{opt}, sp_{i,h}$)}{
        $maxOcc = pathOcc$, $p_{opt} = sp_{i,h}$        \\
        $w_{EL}^d$.removeAll$(w_{EL}^{cover})$ 
        }
Insert $p_{opt}$ into $G_{st}$ by union-find}}
\Return}
\end{algorithm}

From lines 12 to 14, if two paths have the same occurrence score, we obtain distinct edge labels for each of them. A path that includes more edge labels in $w_{EL}$ (line 12) is considered to be more semantically matched to the keyword query. The funnction $compareCover(p_{opt}, sp_{i,h})$ returns true if $sp_{i,h}$ covers more edge labels in $w_{EL}^d$. Once $p_{opt}$ is determined, we apply a union-find algorithm to check whether the candidate path forms a cycle in the current tree or not. Any cyclic paths are eliminated immediately (line 15).

\begin{exmp}
In Figure \ref{fig:steiner_tree_construction} (b), we sort the path candidates by their overall occurrence (\ie $pathOcc = \sum_j occ_j), \forall j \in sp_{i,h}$. For example, the occurrence of the path $(t_0,v_{14}, t_1) = occ_0 + occ_{14} + occ_1 = 3 + 2 + 3 = 8$. Similarly, we have $pathOcc(t_0, v_6, t_1) = 9$, $pathOcc(t_0, v_6, t_2) = 9$, and $pathOcc(t_0, v_7, t_2) = 8$. The paths $(t_0, v_{14}, t_1)$, $(t_0, v_6, t_1)$ start from $t_0$ and end at $t_1$. We keep $(t_0, v_6, t_1)$ instead of $(t_0, v_{14}, t_1)$, since $(t_0, v_6, t_1)$ has higher overall occurrence. Similarly, we keep $(t_0, v_6, t_2)$, while $(t_0, v_7, t_2)$ is abandoned.
\end{exmp}

\begin{figure}[h]
\vspace{-2mm}
\begin{center}
\includegraphics[width=0.8\linewidth]{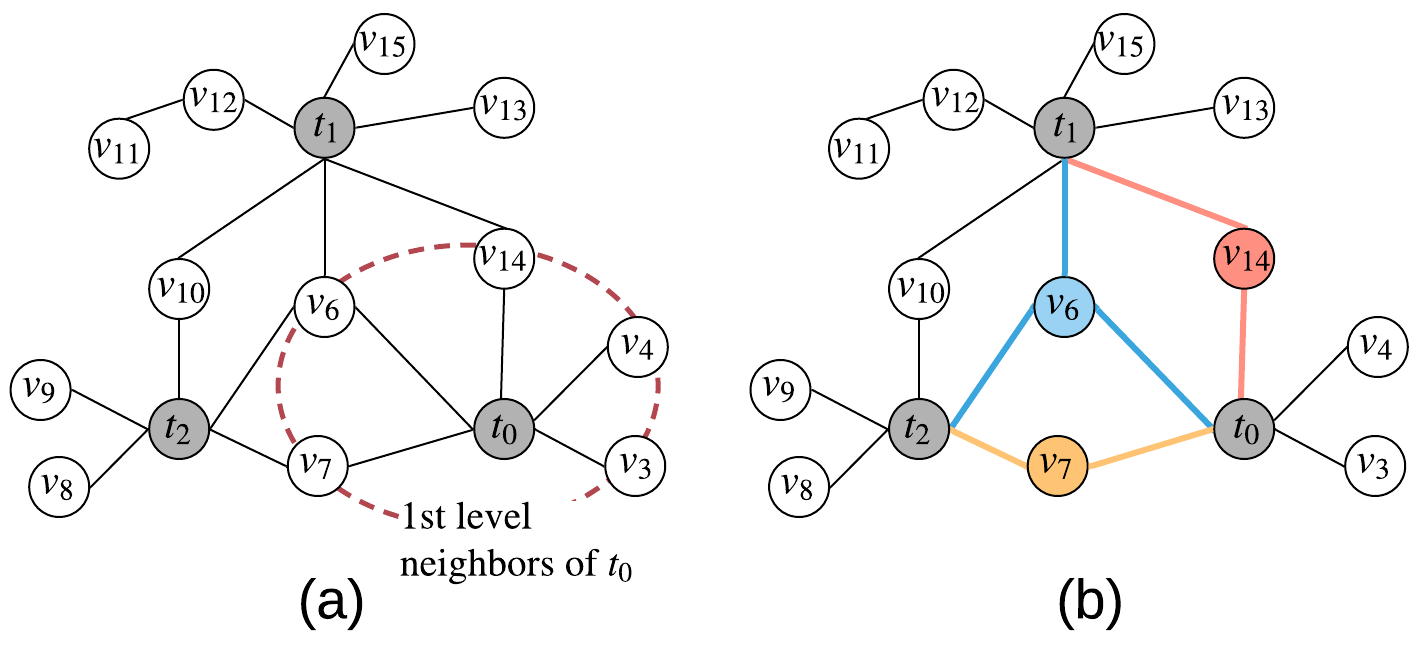}
\captionof{figure}{Examples of ST construction with path selection.}
\label{fig:steiner_tree_construction}
\end{center}
\vspace{-2mm}
\end{figure}

\vspace{1mm}
\subsection{MCS Construction}
\label{subsec:scg-construct}
After the previous steps, we obtain an estimation of an ST $G_{st}$. We use  $G_{st}$ as the backbone to build the MCS by filling it with the edge labels. We next discuss how we handle the dangling edge labels.

For every vertex $v$ in $G_{st}$, we enumerate all the edge labels incident to $v$ in the original graph and compare them to the dangling edge labels. The edge label enumeration is scheduled in a round-robin vertex-by-vertex and level-by-level fashion. All the vertices included in $G_{st}(w)$ are regarded as the initial vertices for exploring the dangling edge labels. 
Then, starting from a vertex $v \in G_{st}$, we search the graph in BFS style, \ie we check the edge labels incident to the current  vertex $v_{curr}$\ignore{ which can be done by looking up the SPO index}. If a dangling edge label $el$ is found to be the label of the edge $(v_{curr}, v'_{curr})$, we retrieve the path from $v$ to $v_{curr}'$ and insert it into $G_{st}(w)$. Otherwise, we search for another vertex $v_{next}$. We repeat this process until all the dangling edge labels have been found.

\begin{figure}[h]
\begin{center}
\includegraphics[width=0.9\linewidth]{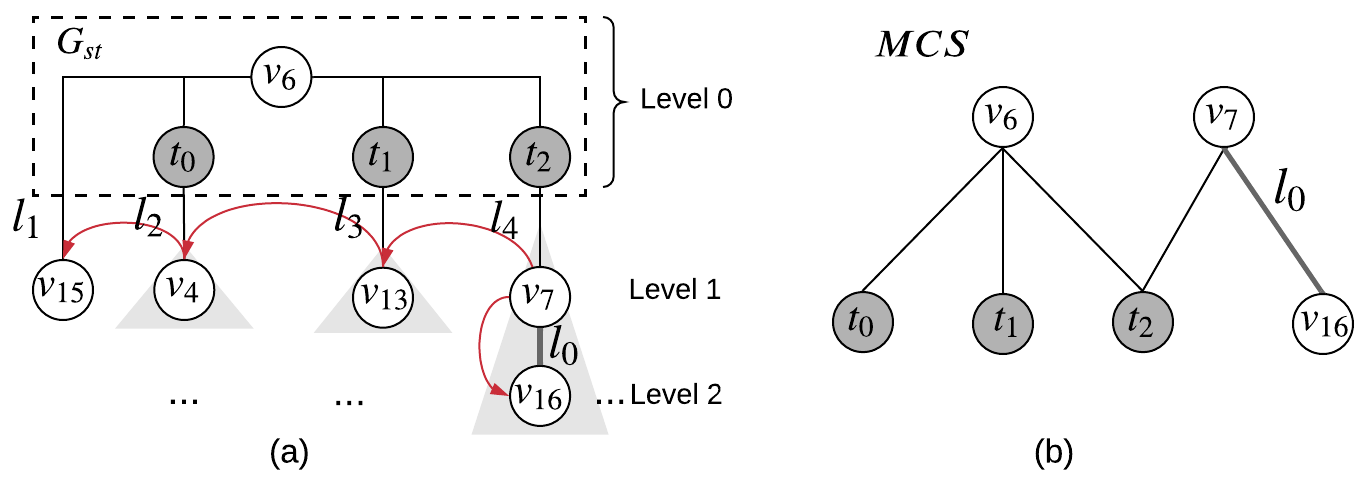}
\captionof{figure}{Finding dangling edge labels. \ignore{Black arrow represents the edge labels has been determined.} Red arrows indicate the search order from one vertex to another.}
\label{fig:dangling_edge}
\end{center}
\end{figure}

\begin{exmp}
In Figure \ref{fig:dangling_edge} (a), we assume $l_0$ is the edge label to be found, and the found ST $G_{st}$ is as shown inside the dotted-line box. We start the process from the vertices at level 0 (that is, the vertices in $G_{st}$), \ie  $t_0, t_1, t_2, v_6$. 
Suppose the order of exploration is $t_2,t_1,t_0,v_6$.   
We first find that the edges incident on $t_2,t_1,t_0,v_6$ do not have label $l_0$. Then, we jump to the next level and check edges incident on $v_7$. Since the edge $(v_7,v_{16})$ has label $l_0$, we retrieve the path from $v_{16}$ to $t_2$ and add it to $G_{st}$.   The final MCS is as shown in Figure \ref{fig:dangling_edge} (b).
\end{exmp}

\vspace{-1mm}
\stitle{Complexity Analysis.}
We use a disk-based skip list to build the \emph{SPO} permutation index, sketches and PLL index. Thus, the search operation on these three indexes can be estimated as $\mathcal{O}(log|V|)$. (1) The time complexity to load the sketches of $k$ keyword vertices is equal to $klog|V|$; (2) For the sketch patch-up, in the worst case, we need to perform $log|V|$ lookups in the SPO index for $k$ keyword vertices, so the complexity is equal to $\mathcal{O}(klog^2|V|)$; (3) In the MCS construction, we need $\mathcal{O}(k^2log^2(|V|))$ lookups in the SPO index to check intersections (between frontiers and neighbor vertices); (4) Finally, when handling dangling edges, the time complexity is $\mathcal{O}(|V + E|)$ (in the worst case, we need to traverse the whole graph). The overall time complexity of MCS construction is thus $\mathcal{O}((1+log|V| + k^2log^2|V|)log|V|)$ (plus $\mathcal{O}(|V + E|)$ in the worst case, when dangling edge labels must be searched over the entire graph).








\section{Ontology Exploration}
\label{sec:ontologyExploration}
We now concentrate on RECON's reasoning mechanism, which employs the Wu-Palmer conceptual similarity \cite{DBLP:conf/acl/WuP94} to measure the closeness between KG concepts. Given a TBox $T$, and two concepts $C_1, C_2 \in T$, the Wu-Palmer conceptual similarity between $C_1$ and $C_2$ is defined as: 

\begin{small}
\begin{equation}
\label{eq:wu-palmer}
    wp(C_1, C_2) = 2*\frac{dep(LCA(C_1,C_2))}{dep(C_1) + dep(C_2)},
\end{equation}
\vspace{-3mm}
\end{small}

\noindent{}where $LCA(C_1, C_2)$ is the lowest common ancestor of $C_1$ and $C_2$, and $dep(C)$ is the depth of concept C in $T$.
$wp(C_1, C_2)$ belongs to $[0, 1]$, and $wp(C_1, C_2) = 1$ if $C_1$ and $C_2$ are identical. Note that if the hierarchy of concepts forms a forest, we will add a pseudo-root $r_p$ as the parent root for all tree roots. Moreover, if the ontology is cyclic, all the concepts in the same cycle are equivalent. Thus, the depth of a given concept is obtained by computing the depth of its corresponding largest strongly connected component \cite{DBLP:journals/semweb/CalvaneseCKKLRR17}.

\begin{defn}[Derivative of Keyword Set]
Given a TBox $T$ and a keyword set $w = \{t_1, t_2, ..., t_n\}$, we say that $w' = \{t_1', t_2',..., t_n'\}$ is a derivative of $w$, if and only if:
$\forall i \in \{1,...,n\}$ we have $t_i = t_i'$ or $t'_i \sqsubseteq t_i$ (\ie  $t'_i$ is a descendant of $t_i$ in $T$).
\end{defn}
\vspace{-3mm}

\begin{exmp}
In Figure \ref{fig:reasoning_semantic}, if $w=\{$dbr:BMW\_M10, dbo:Engine, dbo:predecessor$\}$, $w'=\{$dbr:BMW\_M10, dbo:RocketEngine, dbo:predecessor$\}$ is a derivative of $w$.
\end{exmp}

\stitle{Similarity of Keyword Sets.} 
Consider keyword set $w$ and its derivative $w'$.  Assume $w = \{t_1, t_2,...,t_k,t_{k+1},...t_n\}$, $w' = \{t'_1, t'_2,...,t'_k,t_{k+1},...t_n\}$, where $t'_i\sqsubseteq t_i$.  
\ignore{
Then $LCA(t_i,t_i') = t_i$, and we have
\begin{align}
\label{eq:sim-kw}
Sim(t_i, t_i') = 2*\frac{dep(t_i)}{dep(t_i) + dep(t_i')}.
\end{align}
}
Regarding $w$ and $w'$ as two finite sets, we can calculate their Jaccard similarity: $J(w, w') = \frac{|w\cap w'|}{|w \cup w'|}$. Now since each pair of corresponding elements $(t_i, t'_i)$ in $w$ and $w'$ has a Wu-Palmer similarity $wp(t_i, t'_i)$, we can combine them with the Jaccard similarity $J(w, w')$ to evaluate the similarity between the two keyword sets $w$ and $w'$ as follows:  

\begin{small}
\begin{align}
    Sim(w, w') &= \frac{ |w \cap w'| + (wp(t_1, t_1') + ... + wp(t_k, t_k'))}{|w \cup w'|} \\
    &=\frac{ (n-k) + \sum_{1 \le i \le k} wp(t_i, t_i')}{(n + k)}.
\end{align}
\end{small}
In particular, when $k = 0$, \ie $w$ and $w'$ are identical, we have $Sim(w, w') = 1$.

Algorithm~\ref{algo:reasoning_fwk} shows the ontology exploration process.
We initialize a priority queue of the keyword set with the similarity scores (line 1). We enqueue and sort all derivatives of $w$ by their similarity. The derivative with more similarity from $w$ has a higher priority. We dequeue the keyword set derivatives from the priority queue and compute the corresponding MCS (lines 3 to 4). If an MCS exists, we stop dequeuing further keyword sets, generate and rewrite the SPARQL query $q^{w'}$ for the current keyword set (lines 5 to 7). Finally, we execute $q^{w'}$ in the triple store (line 8). All the other keyword set derivatives which have the same similarity as the chosen keyword set will participate in the SPARQL query rewriting. For example, consider $w = w_0$ and its derivative stored in $PQ$:

\begin{small}
\vspace{-3mm}
\begin{align*}
      PQ = \{
      {w},
      \overbrace{
      \underbrace{w_1, ...,w_{i-1}}_{MCS = \emptyset},
      \underbrace{w_i}_{MCS \neq \emptyset},
      \underbrace{w_{i+1},...,w_j}_{\text{Rewrite}}
      }^{\text{Same similarity to } w} ,w_{j+1},...
   \},
\end{align*}
\end{small}

\noindent{}where $\{w_1,...,w_j\}$ have the same similarity as $w$. We assume that RECON finds an empty MCS for $\{w_1,...,w_{i - 1}\}$, and successfully constructs a non-empty MCS for $w_i$. Then, it generates a SPARQL query based on $w_i$ and stops dequeuing the elements from $PQ$. The system generates the graph patterns of $\{w_{i+1},...,w_j\}$ based on the MCS of $w_i$, and chains them together by UNION operators.

Theoretically, Algorithm \ref{algo:reasoning_fwk} could be very expensive if the number of  subconcepts (\ie descendants) is large. In practice, however, users typically try to provide concrete queries to avoid ambiguity. The statistics we have for DBPedia on the LC-QUAD benchmark show that most types of vertices ($ \geq 95\%$) in a keyword/SPARQL query have less than three subtypes.

\begin{algorithm}[ht!]
\small
\caption{Reasoning Framework}
\label{algo:reasoning_fwk}
\KwIn{$KG = G \cup T$, keyword set $w = w_V \cup w_{EL}$}
\KwOut{Answers of $w$ on $G$}

Initialize $PQ$ with $w$ and the derivatives of $w$ \\

\While{$PQ$ not empty}{
$w' \gets$ dequeue keyword combination from $PQ$ \\
$MCS(w') \gets$ computeMCS($w', T, G$) \\
\If{$MCS(w') \neq \emptyset$}{
    $q^{w'} \gets$ generateSPARQL($MCS(w')$)  \\
    $q^{w'} \gets$ rewriteSPARQL($MCS(w'),T$)  \\
    execute $q^{w'}$ in triple store \\
    \Return answers of $w'$ on $G$ \\
}
}
\end{algorithm}
\vspace{-2mm}

\stitle{Stop Condition.} In Algorithm \ref{algo:reasoning_fwk}, the system stops searching once it finds a non-empty MCS (lines 5 to 9). This guarantees that the derived keyword set used for generating the SPARQL query is the most semantically similar to the original keyword set. In contrast, even if we can find more non-empty trees by continuing to search in the queue, it may result in totally different semantics from the user's intention. For example, given the input keywords \{\texttt{dbo:Person, dbr:Apollo\_11}\}, RECON can find a non-empty MCS via \{\texttt{dbo:Astronaut, dbr:Apollo\_11}\}. However, another MCS with no practical meaning could also be generated by the more concrete keywords combination \{\texttt{dbr:Apollo\_11, dbo:RacingDriver}\}.

\begin{exmp}
Figure \ref{fig:query_generation} gives an example question from DBpedia: ``What is the predecessor of the car engine which was succeeded by BMW M43?" with the keyword query \small{\{\texttt{dbo:Engine, dbo:successor, dbo:predecessor, dbr:BMW\_M43}\}}. The derivatives of the original keyword set are \small{\texttt{\{\{dbo:RocketEngine, dbo:successor, dbo:predecessor, dbr:BMW\_M43\}, \{dbo:AutomobileEngine, dbo:successor, dbo:pre-\\decessor, dbr:BMW\_M43\}\}}}. The input keywords \small{\texttt{\{dbo:-\\engine, dbo:successor, dbo:predecessor, dbr:-\\BMW\_M43\}}} will be considered first for MCS construction. If the system finds a non-empty ST, we generate the SPARQL query, and rewrite it based on the ontology.
\end{exmp}

\begin{figure}[h]
\vspace{-2mm}
\begin{center}
\includegraphics[width=1\linewidth]{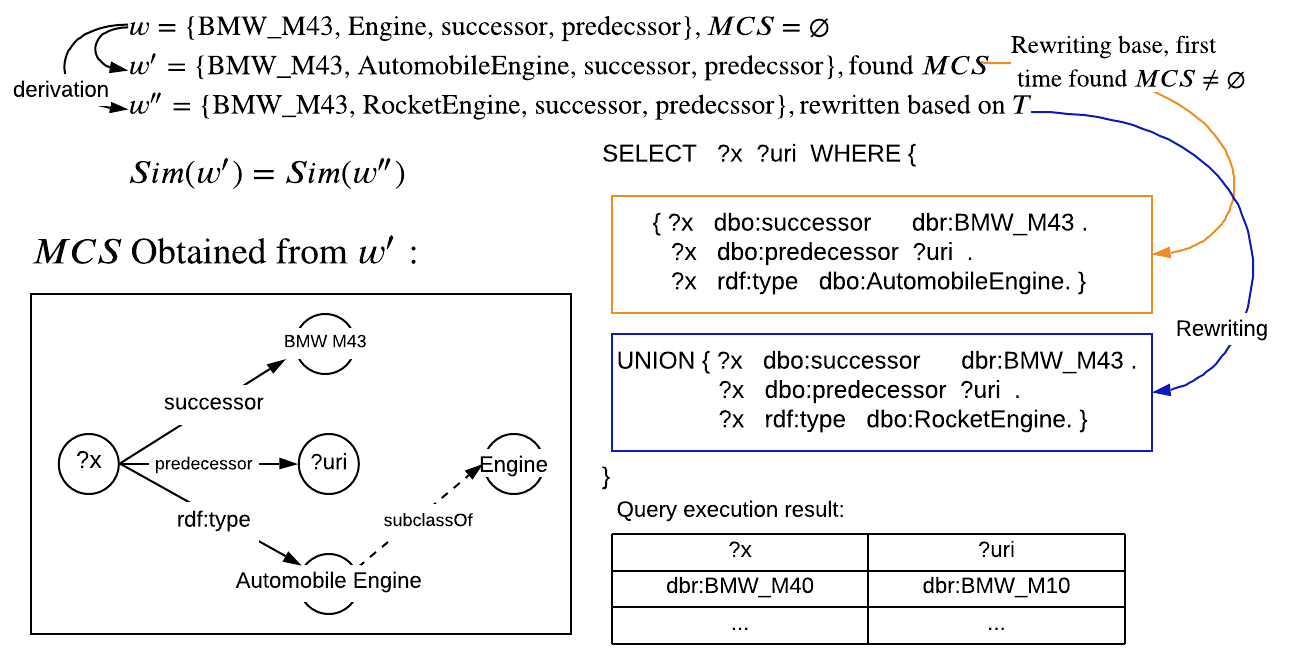}
\captionof{figure}{Example of query generation and rewriting}
\label{fig:query_generation}
\end{center}
\end{figure}

\stitle{Complexity Analysis.}  For a given set of keyword vertices $w = \{t_1, t_2,...,t_n\}$, $W'$ is the set of derivatives from $w$, the size of which depends on the width and depth of the subtree of $t_i \in w$. $T_S$ is the time complexity to construct the MCS for $w$ (\cf Section \ref{subsec:scg-construct}). Thus, the worst-case time complexity of reasoning over $w$ is $|W'|T_S$.

\section{Experiments}
\label{sec:experiments}
The purpose of our experiments is to evaluate RECON in terms of index construction time, index size, online execution time, result accuracy, and the time cost and effectiveness of the reasoning component.
Since there are no previous algorithms for finding an MCS, and approximate ST computation is a major step in our approach, we divide our experiments into two parts: (A) Comparison with existing algorithms for approximate ST construction, and (B) time cost of approximate MCS computation and MCS size. All the experiments are conducted on a server with 56-Core-CPU of Intel(R) Xeon(R) E5-2698v4@2.2GHz, 64 GB RAM, 512 GB disk. For ease of presentation, we will use ST (resp. MCS) to mean approximate ST (resp. approximate MCS) unless explicitly stated otherwise. 

\subsection{Experiments on ST Computation}
In this part of the experiment, we compare RECON with five representative algorithms for ST computation: BANKS II \cite{DBLP:conf/vldb/KacholiaPCSDK05}, BLINKS \cite{DBLP:conf/sigmod/HeWYY07}, DPBF \cite{DBLP:conf/icde/DingYWQZL07}, SketchLS \cite{DBLP:conf/cikm/GubichevN12}, and KeyKG$^+$ \cite{DBLP:conf/www/Shi0K20}. BANKS II, BLINKS, and keyKG{$^+$} are keyword search systems that return approximate group STs; SketchLS is an algorithm for computing an approximate ST; and DPBF is an algorithm for computing the exact group ST. See Section~\ref{sec:relatedWork} for more details on these systems. We implemented all these systems in Java 8. We encoded the input RDF triples into integers with Apache Spark 2.4.x. The SPO permutation and sketch index in RECON are implemented using Apache Lucene 6.x. We use RDF-3X \cite{DBLP:journals/pvldb/NeumannW08} as the triple store (built on gcc 7) in RECON for SPARQL query evaluation.

\stitle{Datasets and Queries.} We use three large-scale real-life datasets and one synthetic dataset for the experiments. DBpedia \cite{DBLP:conf/semweb/AuerBKLCI07}, Wikidata \cite{DBLP:conf/www/Vrandecic12} and Freebase \cite{DBLP:conf/sigmod/BollackerEPST08} are three popular open-domain knowledge bases. The Lehigh University Benchmark (LUBM-$N$) \cite{DBLP:journals/ws/GuoPH05} is a widely used dataset which models the university domain, including instances of universities, departments, students, faculty members, \etc. It is configurable and the parameter $N$ indicates the number of universities (in our case, $N=2000$). We randomly generate 200 queries with the number of vertex keywords $k=2,4,6,8$ for each graph (\ie 1600 queries in total).

Since some systems are not able to handle graphs at such scale (\ie BLINKS index exceeds the memory limit, and SketchLS cannot finish the preprocessing within 24 hours, DPBF takes hours or even days to answer a query), we sample a small subgraph from each of the original datasets to support the evaluation, where the number of edges are set to be $\simeq 100K$ (for LUBM $N$, we use $N=1$). The statistics of these graphs are shown in Table \ref{tab-stat}.


\vspace{-2mm}
\begin{table}[h] \centering
\small
\setlength\tabcolsep{4pt} 
\begin{tabular}{@{}lcccc@{}}\toprule
      &  DBpedia   & Wikidata   & Freebase     & LUBM           \\ 
      \hline 
LG ($|V|/|E|$ in M)    &  49/297  & 132/335  &  97/321  & 66/277   \\
SG ($|V|/|E|$ in K)  &   21/102  & 88/104  &  41/103   & 26/103   \\
\bottomrule
\end{tabular}
\caption{The statistics of the graphs used in the experiments (LG: large graphs; SG: small graphs; M: millions; K: thousands). }
\label{tab-stat}
\end{table}

\stitle{Evaluation Metrics.} We compare RECON with other systems in terms of online execution time (which is the total time for searching and recovering answers), result quality, index construction time and index size. We also evaluate the effect of our path selection and sketch patch-up strategies on the accuracy of the approximate ST, as well as the impact of reasoning on the query time. Considering result quality, we use the following two metrics.

\sstitle{Approximation Error.} We adopt \emph{approximation error} ($App.Er$) \cite{DBLP:conf/cikm/GubichevBSW10,DBLP:conf/cikm/GubichevN12} as our metric for result accuracy: $App.Er =\frac{|ST| - |ST_{min}|}{|ST_{min}|}$ where $|ST|$ is the size of the ST returned by the current system, and $|ST_{min}|$ is the minimum size of the STs returned by all systems under evaluation.

\sstitle{Result Coverage.} Similar to recall, the {\em result coverage} is used to measure the number of smallest STs returned by the current system against that returned by the best among all those under consideration. Let $S_{min}$ be the size of the smallest ST returned by all systems. Note that if an exact algorithm is included, then $S_{sim}$ is the size of the real ST.
The \emph{result coverage} is defined as $RC = \frac{|ST_{curr}|}{max \{|ST_{any}|\}}$, where $|ST_{curr}|$ is the number of STs of size $S_{min}$ returned by the current system, and $|ST_{any}|$ is the number of STs of size $S_{min}$ returned by any system under consideration. Due to computational time constraints over large datasets, we set a maximum value of 10 for $max\{|ST_{any}|\}$. That is, as long as a system returns 10 minimum size STs we stop the search and regard its result coverage to be 1.

\subsection{Experimental Results}
\stitle{Index Size and Construction Time.}
We compare the index size and construction time of BLINKS, SketchLS, KeyKG$^+$ and RECON. For fairness, we set the maximum number of hops (that BFS can reach) to 3 for all systems. The number of partitions in BLINKS is set to 100. We use the implementation of graph betweenness centrality computation from \cite{DBLP:conf/ssdbm/AlGhamdiJSK17} for KeyKG$^+$. Since BLINKS, SketchLS and KeyKG$^+$ do not scale to large graphs, we use LUBM-1 and a medium-size subset of DBPedia (9.4M vertices, 30M edges) for this test. 

Table \ref{tab:offline-preproc} gives the index construction time and index size of the above-mentioned systems. On LUBM-1, RECON is almost an order of magnitude faster than other systems for index construction. Due to the costly computation of the betweenness centrality, KeyKG$^+$ has the longest preproccessing time. DBPedia, BLINKS and KeyKG$^+$ fail to finish the index construction within 24 hours on the medium-size graph. SketchLS takes more than seven hours, which is far longer than RECON (around 30 minutes). Note that SketchLS produces a lot of repeated sketches due to the fact that the same landmarks may be chosen multiple times. However, the RECON index is several times larger than that of SketchLS, and this is mainly due to PLL index which takes 1863 MB. A possible way to reduce the index storage is by compression \cite{DBLP:conf/sigmod/LiQQ0CL19}.

\begin{table}[h] \centering
\small
\setlength\tabcolsep{3.75pt} 
\begin{tabular}{@{}lcccc|cccc@{}}\toprule
&   \multicolumn{4}{c}{\textbf{LUBM-1}}  & \multicolumn{4}{c}{\textbf{DBpedia} }\\ \cmidrule{2-5}\cmidrule{6-9}\\
    &  Bl &  Sk  &  Ky   & \textbf{Re}      &  Bl     &  Sk    &  Ky  &   \textbf{Re} \\ \midrule
 Index Size (MB)    &  4.5  &  0.73  &  6.2   &   5.7    &  -   & 483&  -    &  2200        \\
 Time (s)      & 22.3    & 29.7    &  71   &   2.9   &  -  &  26121  &  -      &  1829  \\
\bottomrule
\end{tabular}
\caption{Index size and offline preprocessing time. Bl: BLINKS; Sk: Sketch; Ky: KeyKG$^+$; Re: RECON. }
\label{tab:offline-preproc}
\end{table}

\stitle{Online Execution Time.} In Figure \ref{fig:big-time}, we plot the execution time of BANKS II, DPBF, KeyKG$^{+}$ and RECON on the original graphs. We also provide the average execution time in Table \ref{tab:avg-et-prec} (upper part). DPBF fails to return the result before the timeout of five minutes for all queries. Blinks and SketchLS cannot construct their indexes within 24 hours.
BANKS II handles most queries quite well and returns the result with a subsecond-level delay. With an increasing number of keywords, BANKS II shows a quasi-linear increase in terms of runtime. For RECON, the execution time stays between $45 \sim 271$ ms, where the disk I/O is the main performance bottleneck. The most costly part comes from PLL path retrieval for online patch-up. For some vertices with a very high degree, it is very expensive to read and deserialize the landmark index into memory and compute the corresponding shortest paths (to minimize memory consumption, RECON stores the indexes on the disk). To understand the performance impact of such behavior, we implement a \emph{hot cache} version of RECON, \ie the sketch and PLL index are pre-loaded into the memory. Consequently, RECON gains a \textbf{$5 \sim 7$ times speedup} for query execution. For KeyKG$^+$, the original version which uses betweenness ordering in constructing the PLL, fails to generate the index within 24 hours on the original graphs. To determine its search efficiency and quality, we use a modified version which constructs the PLL index based on degree ordering. As shown in the table, the modified version of keyKG$^+$ is about five times faster than RECON. This is mainly due to the fact that keyKG$^+$ does not suffer from I/O cost as it stores the index in memory. Compared with the hot cache version of RECON, however, its speed advantage is minor. More importantly, the result accuracy of keyKG$^+$ is unacceptably low, as can be seen in Table \ref{tab:avg-et-prec}.

The execution times on small graphs are plotted in Figure \ref{fig:small-time}. The average execution time are also summarized in Table \ref{tab:avg-et-prec} (lower part). With the increase in the number of keywords, the DPBF runtime increases exponentially. Specifically, to find the top-1 answer, the system requires $\le 1$ second, $3 \sim 4$ minutes, $5$ hours, and $\ge 24$ hours for $k=2,4,6,8$, respectively. Compared with RECON, SketchLS has a slight advantage in search efficiency since it does not have the additional cost of online patch-up, iteration order determination, and path selection. For BLINKS, instead of using Dijkstra's algorithm, we reuse our PLL implementation to compute the shortest paths. BLINKS is much faster than BANKS II, and comparable to RECON in execution time. Note that the high efficiency of BLINKS comes at the cost of large runtime memory storage requirements.

\begin{figure*}[ht]
\centering\captionsetup{format = hang}
\includegraphics[width=1\textwidth]{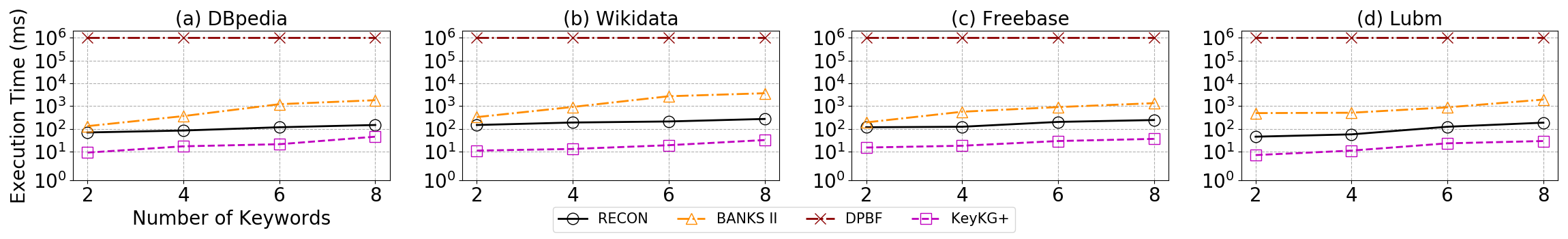}
  \caption{Execution time of RECON, KeyKG$^{+}$ and DPBF, BANKS II on large graphs.}
    \label{fig:big-time}
    \vspace{7mm}

\includegraphics[width=1\textwidth]{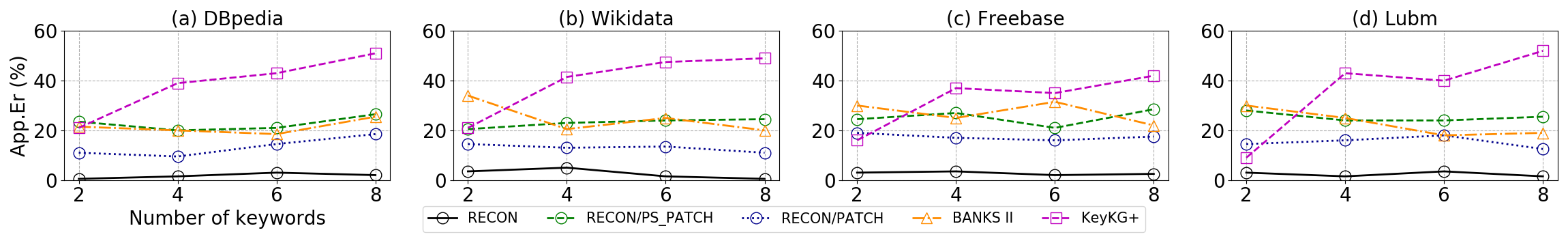}
  \caption{$App.Er$ of ST construction for BANKS II, KeyKG$^{+}$ and RECON on large graphs. RECON/PS\_PATCH: RECON without path selection and patch-up; RECON/PATCH: RECON without patch-up}
    \label{fig:big-precision}
  \vspace{7mm}
  
\includegraphics[width=1\textwidth]{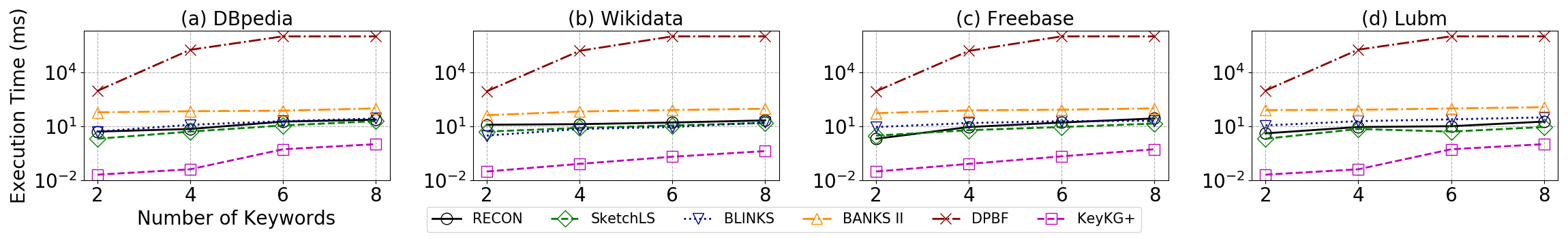}
  \caption{Execution time of RECON, DPBF, SketchLS, BLINKS, BANKS II and KeyKG$^+$ on small graphs.}
    \label{fig:small-time}
  \vspace{7mm}

\includegraphics[width=1\textwidth]{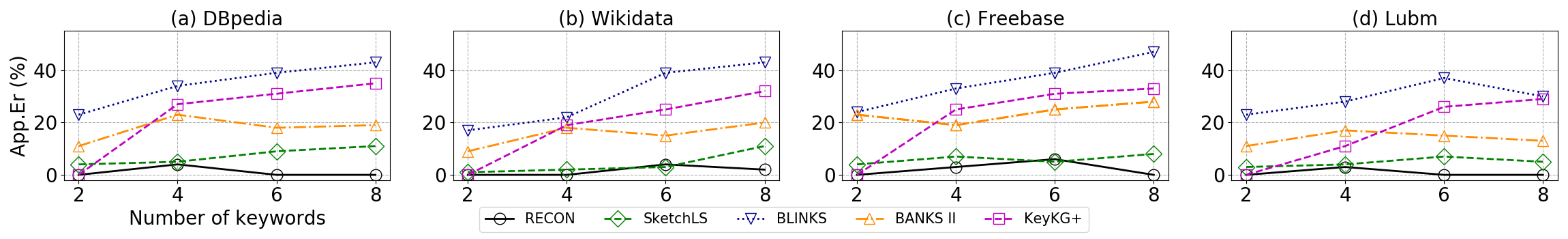}
\caption{$App.Er$ of ST construction for RECON, DPBF, SketchLS, BLINKS, BANKS and KeyKG$^+$ on small graphs.}
  \label{fig:small-precision}
\end{figure*}

\begin{table}[h] \centering
\small
\setlength\tabcolsep{3pt} 
\begin{tabular}{@{}lccccc|ccccc@{}}\toprule
&   \multicolumn{5}{c}{\textbf{ET (ms)}}  & \multicolumn{5}{c}{\textbf{App.Er} (\%)} \\ \cmidrule{2-6}\cmidrule{7-11}
          & Ba   &  Bl &  Sk  &  Ky   & \textbf{Re} &  Ba      &  Bl     &  Sk    &  Ky  &   \textbf{Re}     \\ \midrule
 DBpedia  & 878  &  -  &  -   &  23   &   104       &  21.37   &  -   &  -  & 38.5   &   1.75   \\
 Wikidata & 1899 &  -  & -    &  25   &   205       &  24.87   &  -   &  -  & 39.75  &   2.75  \\
 Freebase & 749  &  -  &  -   &  35   &   171       &  27.12   &  -   &  -  &  35     &   2.75   \\
 LUBM-1   & 951  &  -  &  -   &  18   &   103       &  23      &  -   &  -  &  36.4   &   2.5   \\
\toprule
 DBpedia  & 75  &  16  &  6   &  0.40   &   13      &  17.7    &  34.8   &  7.2  &  23.2  &   1.0   \\
 Wikidata & 70  &  9   &  10  &  0.18   &   16      &  15.5    &  30.2   &  4.2  &  19    &   1.5  \\
 Freebase & 78  &  16  &  8   &  0.21   &   14      &  23.7    &  35.7   &  6.0  &  22.2  &   2.5   \\
 LUBM-1   & 93  &  21  &  6   &  0.32   &   10      &  14.0    &  29.5   &  4.7  &  16.5 &   0.7   \\
\bottomrule
\end{tabular}
\caption{Average execution time and approximation error on \textbf{large graphs} (upper side) and \textbf{small graphs} (lower side). Bl: BLINKS; Sk: Sketch; Ba: BANKS; Ky: KeyKG$^+$; Re: RECON.}
\label{tab:avg-et-prec}
\end{table}
\vspace{3mm}

\stitle{Approximation Error.} Figure \ref{fig:big-precision} presents the $App.Er$ of the result returned by BANKS II, KeyKG$^{+}$ and RECON. In most cases, RECON can return a more compact tree than the other systems. In Table \ref{tab:avg-et-prec} (upper part), we summarize the average $App.Er$ for RECON, KeyKG$^{+}$ and BANKS II. 

Figure \ref{fig:small-precision} shows the $App.Er$ on small graphs. BLINKS has a higher approximation error than the other systems. We do not implement all the optimization techniques from \cite{DBLP:conf/sigmod/HeWYY07} due to missing details (\eg Metis partitioning, batch expansion and scoring function). We observe that the result quality of BLINKS highly depends on the graph partitioning. In fact, the graph partition impacts the search order, causing several of less compact trees may be revealed before the more compact answers. 
In some unfavorable cases, the system searches vertices that are isolated in certain partitions. This causes a lot of cross-partition searches. Such behavior also hinders the search efficiency. RECON and SketchLS both achieve low $App.Er$ in the experiments, while RECON outperforms SketchLS for all datasets. 

\stitle{Effect of path selection and sketch patch-up on approximation error.}
Figure \ref{fig:big-precision} also shows the effect of the path selection strategy and sketch patch-up of RECON on $App.Er$. RECON gains around 9\% and 12\% reduction in $App.Er$ (on average) respectively, after using these techniques.

\begin{table}[h] \centering
\small
\setlength\tabcolsep{5.25pt} 
\begin{tabular}{@{}lccccc|ccc@{}}\toprule
&   \multicolumn{5}{c}{\textbf{SG. $RC$}}   & \multicolumn{3}{c}{\textbf{LG. $RC$}}\\
 \cmidrule{2-6}\cmidrule{7-9}
          & Bl & Sk  & Ba & Ky & \textbf{Re} &  Ba & Ky & \textbf{Re}          \\ \midrule
DBpedia   &  0.56    &  0.69   &  0.78  &  0.60  & 0.75   &  0.91  & 0.37 & 0.75    \\
Wikidata  &  0.41    &  0.71   &  0.89  &  0.58  & 0.83   &  0.82  & 0.45 & 0.80     \\
Freebase  &  0.51    &  0.63   &  0.79  &  0.63  & 0.73   &  0.87  & 0.31 & 0.78     \\
LUBM-1    &  0.58    &  0.71   &  0.83  &  0.57  & 0.79   &  0.89  & 0.41 & 0.82      \\
\bottomrule
\end{tabular}
\caption{Result coverage of all systems evaluated.}
\label{tab:coverage}
\vspace{1mm}
\end{table}

\stitle{Result Coverage.} Table \ref{tab:coverage} summarizes the result coverage $R$ of the systems. In general, BANKS II obtains the highest $R$. This result is not surprising. BANKS II always tries to explore as many answers as possible, causing it to return a lot of irrelevant answers along with the good ones. The result coverage of BLINKS is found to be lower than the other systems. We observe that on small graphs RECON achieves only slightly lower $R$ than BANKS II, where KeyKG$^+$ misses many small STs and has a low result coverage.

\stitle{Reasoning.} We evaluate the impact of reasoning on query execution time and result coverage for RECON. We generate 200 queries from LUBM-2000 (with LUBM ontology as the TBox), where each query $w$ contains at least one concept, which has at least one subclass. We ensure that RECON produces an empty answer for $w$ without reasoning. 

Our experiments indicate that RECON takes 814 ms to execute a query on average after activating reasoning, which is about seven times longer than without reasoning. This is mainly due to the repeated trial with derived keyword sets. Also, by activating the reasoning ability, RECON can find answers for all the queries and gain a result coverage of 1.

\vspace{-2mm}
\subsection{Experiments on MCS Construction } 

\begin{wraptable}{l}{45mm} 
\setlength\tabcolsep{4pt} 
\begin{tabular}{@{}lcccc@{}}\toprule
$|w_{EL}|$   &  $0$     &  $1$   & $2$    & $3$  \\ \hline 
ET (ms)    &  90   &  99   &  105    &  107 \\
\bottomrule
\end{tabular}
\caption{Execution time varies with the number of edge labels for MCS construction.}
\label{tab-mcs-res}
\vspace{2mm}
\end{wraptable}
In this part of the experiment, we compare the time costs for computing ST and MCSs using DBPedia and 200 queries mapped from the LC-QUAD benchmark. The number of vertex and edge label keywords vary from $2 \sim 3$ and $1 \sim 3$, respectively.

We give the average execution time of RECON in Table \ref{tab-mcs-res}. Since the diameters of MCSs are small in LC-QUAD ($\le 2$), the insertion of dangling edge labels is not costly. As shown in Table \ref{tab-mcs-res}, with the growth in the number of keywords, the execution time does not increase significantly. However, this computation could be very expensive in the worst case: for a dangling edge label $el$ and an ST, RECON may have to search the entire graph to find the path containing $el$ in order to build the MCS.

\section{Related Work}
\label{sec:relatedWork}

This work is closely related to (Group) Steiner tree (G/ST) computation and keyword search over KGs.  In Table \ref{tab:cmp}, we summarize the input/output formats, index requirements and underlying semantics for the discussed systems.

\stitle{(G)ST Computation/Approximation.}
Representative works on approximate ST computation include STAR \cite{DBLP:conf/icde/KasneciRSSW09} and SketchLS \cite{DBLP:conf/cikm/GubichevN12}. STAR first builds an initial tree by finding a common ancestor of all input vertices, and then iteratively improves the tree by replacing the longest path in the tree with a shorter one based on several heuristics. SketchLS computes an approximate ST using a sketch-based index and local search, and it has been shown to be faster and more accurate than STAR. However, the offline sketch computation uses BFS from the landmarks across the entire graph, which is too expensive for large-scale graphs like DBPedia. For instance, even with large memory allocation and parallel settings, the offline sketch computation takes more than seven hours on a graph with three million vertices in our experiments.

\begin{table}[ht] \centering
\small
\setlength\tabcolsep{2pt}
\begin{tabular}{@{}l|c|c|c|c|c|c|c|c|c|c|c@{}}\toprule

        & Ba & Bl  & Ea  & St  & Dp & Pd &  Su & Sk & As & Wi & Ky                  \\ \midrule
Input  &  G/R & G &  G/R & G   & G  & G &  RDF & G & RDF & G & G\\
Output &  T/R & T &  G & T   & T  & T &  G & T & G & G & T\\
Index &  - & $\surd$  & $\surd$  & - & - & - & $\surd$ & $\surd$ & $\surd$ & -  & $\surd$\\
SAT  &  GST & - &  - & ST & GST & GST & - & ST & - & - & GST\\
\bottomrule
\end{tabular}
\vspace{3mm}
\caption{Summary of input, output, index computation and underlying semantics for discussed systems. We use the following acronyms: BANKS/BANKS II (Ba), BLINKS (Bl), EASE (Ea), STAR (St), DPBF (Dp), Pd (PrunedDP), \cite{DBLP:journals/tkde/LeLKD14} (Su), SketchLS (Sk), \cite{DBLP:conf/cikm/Han0YZ17} (As), Wikisearch (Wi), KeyKG$^{+}$(Ky). \textbf{Input \& output:} Graph (G), Relational (R). \textbf{Output:} Relational (R), Tree (T), G (Graph). \textbf{SAT:} Group/Steiner-tree-like semantics (G/ST)}
\label{tab:cmp}
\vspace{3mm}
\end{table}

Several graph keyword search systems are based on a GST where given $k$ sets of vertices, the GST is a minimum connected tree that contains a vertex from each set. Among these systems, BANKS \cite{DBLP:conf/icde/BhalotiaHNCS02} computes an approximate GST with a backward search algorithm, and BANKS II \cite{DBLP:conf/vldb/KacholiaPCSDK05} supports forward search and a vertex-activation-based approach to further accelerate the search process. BLINKS \cite{DBLP:conf/sigmod/HeWYY07} is a keyword search system over general graphs, where the output is a tree with minimum total distance from the root to the keyword vertices. Note that the tree is similar but not identical to a GST. BLINKS introduces a partition-based indexing technique for fast shortest path retrieval. However, the preprocessing of BLINKS runs $N$ copies of Dijkstra's algorithm for all-pairs shortest path computation, which is expensive and consumes a huge amount of memory. Besides, the quality of results returned by BLINKS highly depends on how the graph is partitioned.

DPBF \cite{DBLP:conf/icde/DingYWQZL07} computes an exact (G)ST using parameterized dynamic programming. However, as demonstrated in our experiments, DPBF barely scales to even moderate-sized graphs. PrunedDP$^{++}$ \cite{DBLP:conf/sigmod/LiQYM16} improves DBPF by quickly finding an initial tree, and then progressively refining it until an exact GST is found. The key techniques of PrunedDP$^{++}$ are optimal tree decomposition, conditional tree merging, and A*-search. However, the approximate GSTs returned by PrunedDP$^{++}$ are of low quality, and the computation of exact solutions is too slow for large graphs, as noted in \cite{DBLP:conf/www/Shi0K20}.

KeyKG$^{+}$ \cite{DBLP:conf/www/Shi0K20} is the most recent work on keyword search over knowledge graphs under GST semantics. It first selects a vertex combination, with one vertex corresponding to each keyword, according to their overall closeness, and then generates an approximate ST. In both steps, it uses the PLL index to speed up the shortest distance/path computation. In the computation of the PLL, the vertices are ordered using (approximate) betweenness centrality. However, as shown in our experiments, the accuracy of an ST returned by KeyKG$^+$ is lower than other systems. Besides, the cost of the betweenness computation is very high. For example, it takes more than 24 hours to process a moderate-sized graph with the state-of-the-art implementation from \cite{DBLP:conf/ssdbm/AlGhamdiJSK17}. Moreover, the size of the PLL index grows exponentially in large dense graphs, which makes KeyKG$^+$ inapplicable to large graphs such as DBPedia. 
Besides, KeyKG$^{+}$ assumes that the keywords cannot be mapped to edges, and the authors argue that each edge label can be replaced with a two-edge path going through a vertex with the same label. However, this strategy will significantly increase the graph size, making the computation even slower.

\vspace{-2mm}
\begin{table}[h] \centering
\small
\setlength\tabcolsep{5pt}
\begin{tabular}{@{}lc@{}}\toprule
&   Time Complexity  \\ \bottomrule
BANKS II  & $\mathcal{O}(|V|^2log|V| + |V||E|)$    \\
BLINKS   & Depends on graph partitioning      \\
DPBF  &  $\mathcal{O}((3^k + 2^k(k + log|V|)|V| + |E|))$ \\
SketchLS   & $\mathcal{O}(k^3log^3|V|)$     \\
KeyKG$^{+}$   & $\mathcal{O}(k|V|^2 + k^3|V| )$    \\
RECON  &   $\mathcal{O}((1+log|V| + k^2log^2|V|)log|V|)$ \\
\bottomrule
\end{tabular}
\caption{Time complexity of the systems involved in our experiments. }
\label{tab:offline}
\end{table}
\vspace{2mm}

Compared with the above, our system can find highly accurate approximate STs and scale to large, dense graphs. In Table \ref{tab:offline}, we summarize the time complexity of the systems involved in our experiments. For more details regarding the complexity analysis of RECONO, please rerfer to Section \ref{sec:scg_construction}.

\stitle{Other KG Keyword Search Systems.}
EASE \cite{DBLP:conf/sigmod/LiOFWZ08} proposes a radius-based index to support efficient \emph{Steiner Graph} searching on unstructured, semi-structured and structured datasets. However, \cite{DBLP:journals/pvldb/KargarA11,DBLP:conf/cikm/CoffmanW10,DBLP:conf/cikm/GubichevN12} point out the poor scalability of EASE in graphs with millions of vertices, since the index is based on the computation of a high-order adjacency matrix. \cite{DBLP:journals/tkde/LeLKD14} is a keyword search system over RDF graphs using the same semantics as BLINKs. It generates a type-level summary to accelerate query evaluation. The authors generate SPARQL queries based on the primitive returned query pattern, and use a triple store to support query answering. \cite{DBLP:conf/cikm/Han0YZ17} targets question-answering-style keyword search on large KGs. The authors build a bipartite graph from the original RDF dataset and convert the keyword search problem to a graph assembly problem. The final output is the answers of the generated SPAPRQL query. Wikisearch \cite{DBLP:conf/icde/YangAJTW19} proposes an approach for keyword search over KGs where the search result is a so-called {\em central graph} instead of a GST. To achieve low-latency (sub-second delay) query answering, it exploits the parallelism capability of modern hardware (e.g., GPUs). The query semantics of these systems are different from that of RECON and, to the best of our knowledge, no previous work has exploited ontology to improve search quality.

\vspace{-1mm}
\section{Conclusion}
\label{sec:conclusion}
In this paper, we focus on efficiently finding minimum connected subgraphs over KGs with ontological reasoning support. Moreover, we adopt a landmark-based index with sketch patch-up and path selection to efficiently identify the corresponding MCS of an input keywords query, which represents the backbone of a SPARQL query. Based on a set of comprehensive experiments, we show that our experiments demonstrate that RECON can achieve high accuracy with instantaneous response over large graphs (hundreds of millions edges) for ST computation, and for practical queries, the MCS computation can be done at a small additional cost.

\balance
\bibliographystyle{abbrv}
{\footnotesize\bibliography{main}}

\end{sloppypar}
\end{document}